\documentclass[11pt,reqno,a4paper]{amsart}
\usepackage{fourier}
\usepackage[T1]{fontenc}
\usepackage[utf8]{inputenc}
\usepackage[spanish,english,activeacute]{babel}
\usepackage{mathtools}
\usepackage{amsmath,amssymb,amsthm}
\usepackage{enumerate}
\usepackage[colorlinks=true]{hyperref}
\usepackage{url}
\usepackage{stmaryrd}
\usepackage{faktor}
\usepackage{hyphenat}
\usepackage{tikz}
\usepackage{pgfplots}
\usepackage{color}
\usepackage{caption}
\usepackage{subcaption}
\pgfplotsset{compat=1.17}
\usepackage{multicol}

\usepackage[a4paper,top=3.3cm,bottom=3cm,left=3cm,right=3cm,bindingoffset=5mm]{geometry}

\usepackage{marginnote}
\newtheorem{thm}{Theorem}[section]
\newtheorem{col}[thm]{Corollary}
\newtheorem{prop}[thm]{Proposition}
\newtheorem{lemma}[thm]{Lemma}

\theoremstyle{definition}
	\newtheorem{mydef}[thm]{Definition}
	\newtheorem{ex}[thm]{Example}
	\newtheorem{exs}[thm]{Examples}
	\newtheorem{remark}[thm]{Remark}
	\newtheorem{remarks}[thm]{Remarks}

\numberwithin{equation}{section}

\renewcommand{\thesubfigure}{\arabic{subfigure}}

\DeclareMathOperator{\dis}{d}
\DeclareMathOperator{\supp}{supp}
\DeclareMathOperator{\ev}{ev}
\DeclareMathOperator{\w}{w}
\DeclareMathOperator{\tr}{tr}
\DeclareMathOperator{\btr}{\textbf{tr}}

\DeclarePairedDelimiter\ceil{\lceil}{\rceil}
\DeclarePairedDelimiter\floor{\lfloor}{\rfloor}

\title[Optimal $(r,\delta)$-LRCs from monomial-Cartesian codes and their subfield-subcodes]{Optimal $(r,\delta)$-LRCs from monomial-Cartesian codes and their subfield-subcodes}

\author{C. Galindo}
\email[Carlos Galindo]{galindo@uji.es}
\address{Instituto Universitario de Matemáticas y Aplicaciones de Castellón and Departamento de
Matemáticas, Universitat Jaume I, Campus de Riu Sec., 12071 Castelló, Spain}

\author{F. Hernando}
\email[Fernando Hernando]{carrillf@uji.es}
\address{Instituto Universitario de Matemáticas y Aplicaciones de Castellón and Departamento de
Matemáticas, Universitat Jaume I, Campus de Riu Sec., 12071 Castelló, Spain}

\author{H. Martín-Cruz}
\email[Helena Martín-Cruz]{martinh@uji.es}
\address{Instituto Universitario de Matemáticas y Aplicaciones de Castellón, Universitat Jaume I, Campus de Riu Sec., 12071 Castelló, Spain}

\thanks{This work has been partially supported by
MCIN/AEI/10.13039/501100011033 and by the “European Union NextGenerationEU/PRTR”, grant TED2021-130358B-I00, as well as by Universitat Jaume I, grants UJI-B2021-02, GACUJIMA/2023/06 and PREDOC/2020/39.}

\begin{document}

\begin{abstract}
We study monomial-Cartesian codes (MCCs) which can be regarded as $(r,\delta)$-locally recoverable codes (LRCs). These codes come with a natural bound for their minimum distance and we determine those giving rise to $(r,\delta)$-optimal LRCs for that distance, which are in fact $(r,\delta)$-optimal. A large subfamily of MCCs admits subfield-subcodes with the same parameters of certain optimal MCCs but over smaller supporting fields. This fact allows us to determine infinitely many sets of new $(r,\delta)$-optimal LRCs and their parameters.
\end{abstract}

\maketitle

\section*{Introduction}
\sloppy Locally recoverable (or repairable) codes (LRCs) were introduced in \cite{GHSY2012}. The aim was to consider error-correcting codes to treat the repair problem for large scale distributed and cloud storage systems. Thus an error-correcting code $C$ is named an LRC with locality $r$ whenever any symbol in $C$ can be recovered by accessing at most $r$ other symbols of $C$ (see, for instance, the introduction of \cite{GHM2020} for details). The literature contains a good number of papers on this class of codes, some of them are \cite{ZXL2017,LMC2018,Mi2018,LXY2019,J2019,LMT2020,SVV2021}. A variation of Reed-Solomon codes was introduced in \cite{TB2014} for recovering purposes. In \cite{BTV2017} these codes were extended to LRCs over algebraic curves. Among the different classes of codes considered as good candidates for local recovering, cyclic codes and subfield-subcodes of cyclic codes play an important role, this is because the cyclic shifts of a recovery set again provide recovery sets \cite{CXHF2018,GC2014,HYUS2015,TBGC2015}. In \cite{Mu2018} the author introduces a model of locally recoverable code that also includes local error detection, increasing the security of the recovery system.

There is a Singleton-like bound for LRCs with locality $r$ \cite{GHSY2012}. Codes attaining this bound are named optimal $r$-LRCs and interesting constructions of this class of codes can be found in \cite{TB2014} and \cite{TPD2016} (see also \cite{BHHMV2017,BTV2017,MT2018,MTT2020,SVV2021}). When considering codes over the finite field $\mathbb{F}_q$, $q$ being a prime power, optimal $r$-LRCs can be obtained for all lengths $n\leq q$ \cite{ZY2016} and a challenging question is to study how long these codes can be \cite{GXY2019}.

The fact that simultaneous multiple device failures may happen leads us to the concept of LRCs with locality $(r,\delta)$ (or $(r,\delta)$-LRCs). This class of codes were introduced in \cite{PKLK2012}, see Definition \ref{def rd} in this paper, and they also admit a Singleton-like bound \cite{PKLK2012}, which we reproduce in Proposition \ref{pformula}. Codes attaining this bound are named optimal $(r,\delta)$-LRCs or, in this paper, simply optimal codes. Optimal codes have been studied in \cite{CXHF2018,LMX2019,LXY2019,SZW2019,J2019,CFXF2019,FF2020,QZF2021}, mainly coming from cyclic and constacyclic codes. A somewhat different way for obtaining LRCs with locality $(r,\delta)$ was started in \cite{GHM2020}, where the supporting codes were the so-called $J$-affine variety codes. These codes were introduced in \cite{GHR2015} and they have a good behaviour for constructing quantum error-correcting codes \cite{GH2015,GHR2015,GGHR2017}.

Monomial-Cartesian codes (MCCs) are a class of error-correcting codes, introduced in \cite{LMS2020}, that contains the set of $J$-affine variety codes. They are evaluation codes obtained as the image of maps
$$
\ev_P\colon V_\Delta \subset \faktor{\mathbb{F}_q[X_1,\dots,X_m]}{I} \to \mathbb{F}_q^{n} \textrm{, } \quad \ev_P(f)=\left(f(\boldsymbol{\alpha}_1),\dots,f(\boldsymbol{\alpha}_n)\right),
$$
where $m$ is a positive integer larger than $1$, $P=P_1\times\dots\times P_m=\{\boldsymbol{\alpha}_1,\dots,\boldsymbol{\alpha}_n\}$ a suitable subset of $\mathbb{F}_q^m$, $I$ the vanishing ideal at $P$ of $\mathbb{F}_q[X_1,\dots,X_m]$ and $V_\Delta$ an $\mathbb{F}_q$-linear space generated by classes of monomials (Definition \ref{MCC}). This evaluation map is also used in \cite{CN2021} to define codes with variable locality and availability. Evaluation maps of our codes are defined on subsets of coordinate rings of certain affine varieties, but these codes can also be introduced with algebraic tools, as in \cite{LMS2020}.

The goal of this paper is to obtain many new optimal LRCs coming from MCCs. Previously, an algebraic description of MCCs was given in \cite{LSV2021} and these codes were considered for applications different of those in this paper, such as quantum codes, LRCs with availability and polar codes \cite{LMS2020,CLMS2021}.

MCCs come with a natural bound on their minimum distance which allows us to obtain many optimal $(r,\delta)$-LRCs. In fact, we are able to get all MCCs providing optimal codes whose minimum distance coincides with the mentioned bound (see Remark \ref{nomore}).

MCCs are related with and include the family of codes introduced in \cite{ACNV2021} whose evaluation map is the same as MCCs but their evaluation sets $V_\Delta$ are only a subset of ours. This makes that the sets $\Delta$ in \cite{ACNV2021} have specific shapes while ours can have arbitrary shapes and therefore we obtain many more optimal $(r,\delta)$-LRCs (see Remark \ref{comp} for details).

We are interested in optimal $(r,\delta)$-LRCs and the recent literature presents a number of results giving parameters of codes of this type \cite{CXHF2018,SZW2019,CFXF2019,XC2019,CH2019,FF2020,Z2020,ZL2020,CWLX2021,LEL2021,KWG2021}. The length of most of these codes is a multiple of $r+\delta-1\leq q$ and, in this case, and for unbounded length and small size fields, their distances have restrictions being at most $3\delta$. Larger distances can be obtained when $q^2+q$ is a bound for the length. One must use different constructions to get these optimal codes, and a large size of the supporting field seems to make easier to find optimal codes \cite{SDYL2014}.

MCCs are generated by evaluating monomials in several variables and the set of exponents of their generators determines the dimension and a bound $d_0$ for the minimum distance (see Proposition \ref{lendim} and Corollary \ref{pdist}). Our recovery procedure based on interpolation also makes easy to obtain the values $r$ and $\delta$ of some MCCs regarded as LRCs (Proposition \ref{method}). Supported on these facts, we provide a large family of optimal MCCs. Subsection \ref{subm2} is devoted to bivariate codes and Subsection \ref{subm3} to multivariate codes. In fact, codes given in Propositions \ref{rect}, \ref{rectelim} and \ref{rectwithp}, \ref{hyperrect} and \ref{hyperrectelim} give the $d_0$-optimal LRCs one can get with this type of codes. Notice that $d_0$-optimal codes (Definition \ref{d0op}) are optimal codes by Remark \ref{facts} (2).

The above five propositions determine all the parameters of the $d_0$-optimal LRCs given by MCCs, see Remarks \ref{nomore} and \ref{gnomore}. These parameters are grouped in Corollary \ref{colpar} for the bivariate case and in Corollary \ref{gcolpar} for the multivariate case. Thus, one gets a large family of optimal LRCs that can be constructed by a unique and simple procedure.

This family provides, on the one hand, the parameters of those LRCs over $\mathbb{F}_q$ given in \cite{CH2019} whose lengths are of the form $N(r+\delta-1)$ where $N$ can be written as a product of integers less than or equal to $q$ and, on the other hand, the parameters of those LRCs in \cite{LEL2021} with length less than or equal to $q^2+q$.

The above codes do not give new parameters but subfield-subcodes of many subfamilies of them \textit{do give}. Thus, providing \textit{new} families of optimal LRCs is our main goal. Indeed, in Section \ref{subfield} we prove that, considering suitable subfield-subcodes over subfields $\mathbb{F}_{q'}$ of $\mathbb{F}_q$, we get LRCs over $\mathbb{F}_{q'}$ with the same parameters of certain MCCs over $\mathbb{F}_q$. Propositions \ref{bivq+1} and \ref{bivq+2} for the bivariate case, and Propositions \ref{multq+1} and \ref{multq+2} for the multivariate case explain how to construct new optimal $(r,\delta)$-LRCs.

The main results of the paper are Theorems \ref{thmparbivq+1}, \ref{thmparbivq+2}, \ref{thmparmultq+1} and \ref{thmparmultq+2}. Theorem \ref{thmparbivq+1} (respectively, \ref{thmparmultq+1}) gives parameters of new optimal LRCs over any field coming from the bivariate (respectively, multivariate) case. Theorems \ref{thmparbivq+2} and \ref{thmparmultq+2} do their own but only for characteristic two fields. Remarks \ref{remsub} and \ref{remsub2} justify the novelty of our codes. Finally, in Examples \ref{examplessubfield} and Tables \ref{table:parbiv} and \ref{table:parmult}, one can find some numerical examples of new optimal LRCs over small fields.

Section \ref{defs} of the paper is a brief introduction to locally recoverable codes (LRCs) and monomial-Cartesian codes (MCCs) are introduced in Section \ref{framework} as well as how they can be considered as LRCs, being Proposition \ref{method} the main result in this section. Section \ref{work} is devoted to determine the set of optimal MCCs we can obtain. We divide our study in two cases: bivariate and multivariate performed in Subsections \ref{subm2} and \ref{subm3}. Finally our main results concerning new LRCs obtained from subfield-subocdes of some MCCs are given in Section \ref{subfield}. Subsection \ref{factssubfield} recalls the results on subfield-subcodes we will use, while the new parameters are given in Subsection \ref{opbivsubfield} where the bivariate case is treated and in Subsection \ref{opmultsubfield} devoted to the multivariate case.

\section{Locally recoverable codes}\label{defs}

In this section we give a brief introduction to locally recoverable codes (LRCs). An LRC is an error-correcting code such that any erasure in a coordinate of a codeword can be recovered from a set of other few coordinates. Let $q$ be a prime power and $\mathbb{F}_q$ the finite field with $q$ elements. Let $C$ be a linear code over $\mathbb{F}_q$ with parameters $[n,k,d]_q$. A coordinate $i \in \{1,\dots,n\}$ is \textit{locally recoverable} if there is a \textit{recovery set} $R\subseteq \{1,\dots,n\}$ with cardinality $r>0$ and $i\notin R$ such that for any codeword $\mathbf{c}=(c_1,\dots,c_n)\in C$, an erasure in the coordinate $c_i$ of $\mathbf{c}$ can be recovered from the coordinates of $\mathbf{c}$ with indices in $R$. Set $\pi_R \colon \mathbb{F}_q^n \to \mathbb{F}_q^r$ the projection map on the coordinates of $R$ and write $C[R]:=\{\pi_R(\mathbf{c}) \mid \mathbf{c} \in C\}$. Then:

\begin{prop}\label{rset}
A set $R\subseteq \{1,\dots,n\}$ is a recovery set for a coordinate $i\notin R$ if and only if $\dis(C[\overline R])\geq 2$, where $\overline R=R\cup\{i\}$ and $\dis$ stands for the minimum distance.
\end{prop} 

The \textit{locality} of a coordinate is the smallest cardinality of a recovery set for that coordinate. An \textit{LRC with locality} $r$ is an LRC such that every coordinate is locally recoverable and $r$ is the largest locality of its coordinates. The parameters and locality of an LRC satisfy the following Singleton-like inequality.
$$k+d+\ceil*{\frac{k}{r}}\leq n+2.$$
When the equality holds, the code is called \textit{optimal $r$-LRC}.

By Proposition \ref{rset}, if $R$ is a recovery set for $i$, then $\dis(C[\overline R])\geq 2$ and thus only one erasure can be corrected (also only up to to one error can be detected). But erasures can also occur in $\pi_R(\mathbf{x})$ and then we could not recover $x_i$. To correct more than one erasure we introduce the concept of locality $(r,\delta)$, also named $(r,\delta)$-locality.

\begin{mydef}\label{def rd}
A code $C$ is \textit{locally recoverable with locality} $(r,\delta)$ if, for any coordinate $i$, there exists a set of coordinates $\overline R=\overline R(i) \subseteq \{1,\dots,n\}$ such that:
\begin{enumerate}
    \item $i\in \overline R$ and $\#\overline R \leq r+\delta-1$; and
    \item $\dis(C[\overline R])\geq \delta$.
\end{enumerate}
Such a set $\overline R$ is called an $(r,\delta)$\textit{-recovery set for $i$} and $C$ an $(r,\delta)$\textit{-LRC}.
\end{mydef}

In this paper, we will always refer to this type of locality and sometimes, abusing the notation, we will talk about locality $r$ understanding locality $(r,\delta)$ for some $\delta$ inferred from the context. The second condition in Definition \ref{def rd} allows us to correct an erasure at coordinate $i$ plus any other $\delta-2$ erasures in $\overline R \backslash \{i\}$ by using the remaining $r$ coordinates (also it allows us to detect an error at coordinate $i$ plus any other $\delta-2$ errors in $\overline R \backslash \{i\}$).
Notice that, when $\delta\geq 2$ and $C$ is an LRC with locality $(r,\delta)$, the (original definition of) locality of $C$ is $\leq r$. In fact, any subset $R\subseteq \overline R$ such that $\#R=r$ and $i\notin R$ fulfills $\dis(C([R]\cup \{i\}))\geq 2$, so by Proposition \ref{rset} $R$ is a recovery set for the coordinate $i$. There is also a Singleton-like inequality for $(r,\delta)$-LRCs:

\begin{prop}\label{pformula}\cite{PKLK2012}
The parameters $[n,k,d]_q$ of an $(r,\delta)$-LRC, $C$, satisfy
\begin{equation}\label{formula}
    k+d+\left(\ceil*{\frac{k}{r}}-1\right)(\delta-1)\leq n+1.
\end{equation}
\end{prop}
In this paper, $C$ is called an \textit{optimal $(r,\delta)$-LRC} (or simply, an \textit{optimal} LRC) whenever equality holds in (\ref{formula}).

In the next section we define the linear codes we will use for local recovery.

\section{Monomial-Cartesian codes}\label{framework}

Let $m>1$ be a positive integer and consider a family $\left\{P_j\right\}_{j=1}^m$ of subsets of $\mathbb{F}_q$ with cardinality larger than one. Set
$$P=P_1 \times \dots \times P_m=\{\boldsymbol{\alpha}_1,\dots,\boldsymbol{\alpha}_n\} \subseteq \mathbb{F}_q^m.$$
We usually write $\boldsymbol{\alpha}_i=({\alpha_i}_1,\dots,{\alpha_i}_m)$. Consider the quotient ring
$$\mathcal{R}=\faktor{\mathbb{F}_q[X_1,\dots,X_m]}{I},$$
where $I$ is the ideal of the polynomial ring in $m$ variables $\mathbb{F}_q[X_1,\dots,X_m]$ vanishing at $P$. Then, $I=\langle f_1(X_1),\dots,f_m(X_m)\rangle$, where $f_j(X_j)=\prod_{\beta\in P_j}(X_j-\beta)$ and $\deg(f_j)=\# P_j=:n_j\geq 2$ \cite{LRV2014}. Let
$$E=\{0,1,\dots,n_1-1\}\times\dots\times\{0,1,\dots,n_m-1\}.$$
Given $f\in \mathcal{R}$, $f$ denotes both the equivalence class in $\mathcal{R}$ and the unique polynomial in $\mathbb{F}_q[X_1,\dots,X_m]$ with degree in $X_j$ less than $n_j$, $1\leq j \leq m$, representing $f$. Thus
$$f(X_1,\dots,X_m)=\sum_{(e_1,\dots,e_m)\in E} f_{e_1,\dots,e_m}X_1^{e_1}\cdots X_m^{e_m},$$
with $f_{e_1,\dots,e_m}\in\mathbb{F}_q$. Set $\supp(f)=\{(e_1,\dots,e_m)\in E \mid f_{e_1,\dots,e_m}\neq 0\}$. For each subset $\emptyset\neq\Delta \subseteq E$, define $V_\Delta:=\{f\in \mathcal{R} \mid \supp(f)\subseteq \Delta\}$ and for each element $\mathbf{e}=(e_1,\dots,e_m)\in E$, denote $X^{\mathbf{e}}=X_1^{e_1}\cdots X_m^{e_m}$. Then, $V_\Delta$ is the $\mathbb{F}_q$-vector space $\langle X^{\mathbf{e}} \mid \mathbf{e}\in \Delta\rangle$. The linear evaluation map
$$
\ev_P\colon \mathcal{R} \to \mathbb{F}_q^{n} \textrm{, } \quad \ev_P(f)=\left(f(\boldsymbol{\alpha}_1),\dots,f(\boldsymbol{\alpha}_n)\right),
$$
gives rise to the following class of evaluation codes.

\begin{mydef}\label{MCC}
The \textit{monomial-Cartesian code} (MCC) $C_\Delta^P$ is the following vector subspace of $\mathbb{F}_q^n$ over the finite field $\mathbb{F}_q$:
$$C_\Delta^P:=\ev_P(V_\Delta)=\langle\ev_P(X^{\mathbf{e}}) \mid \mathbf{e}\in \Delta\rangle \subseteq \mathbb{F}_q^n.$$
We say that the MCC $C_\Delta^P$ is bivariate (respectively, multivariate) when $m=2$ (respectively, $m>2$).
\end{mydef}

MCCs were introduced in \cite{LMS2020} in a different way (using only algebraic tools), and are a family of codes that extend $J$-affine variety codes introduced in \cite{GHR2015}. Denoting by $U_t\subseteq \mathbb{F}_q$ the set of $t$-roots of unity for some $t\mid q-1$, a $J$-affine variety code is an MCC where each $P_j$ is of the form $U_t$ or $U_t\cup\{0\}$.

We also introduce the following definition which will be useful in the next sections.

\begin{mydef}
Two subsets $\Delta_1$ and $\Delta_2$ of $E$ are \textit{pseudoisometric} if there exists $\mathbf{v}=(v_1,\dots,v_m)\in \mathbb{Z}^m$ such that
$$\Delta_2=\mathbf{v}+\Delta_1:=\{(e_1+v_1,\dots,e_m+v_m) \mid (e_1,\dots,e_m)\in\Delta_1\}.$$
In that case, we say that the codes $C^P_{\Delta_1}$ and $C^P_{\Delta_2}$ are \textit{pseudoisometric}.
\end{mydef}

\begin{remark}\label{rempseudo}
In this paper, we say that two codes are isometric if there exists a bijective mapping between them that preserves Hamming weights. The * product of two vectors $(v_1,\dots,v_n)$ and $(w_1,\dots,w_n)$ in $\mathbb{F}_q^n$ is defined as:
$$(v_1,\dots,v_n)*(w_1,\dots,w_n)=(v_1\cdot w_1,\dots,v_n\cdot w_n).$$
Then $\ev_P(fg)=\ev_P(f)*\ev_P(g)$ for all $f,g\in \mathcal{R}$.

Assume that $\Delta_1$, $\Delta_2\subseteq E$ are pseudoisometric sets such that $\Delta_2=\mathbf{v}+\Delta_1$ and $v_j\neq 0$ for $1\leq j \leq m$. For simplicity, suppose $v_j<0$, $1\leq j \leq m_1$, and $v_j>0$, $m_1+1\leq j \leq m$. Consider
$$\Delta_2^{'}=(-v_1,-v_2,\dots,-v_{m_1},0,\dots,0)+\Delta_2$$
and
$$\Delta_1^{'}=(0,\dots,0,v_{m_1+1},\dots,v_m)+\Delta_1,$$
and then $\Delta_2^{'}=\Delta_1^{'}$. Thus
$$V_{\Delta_2^{'}}=\left\{X_1^{-v_1}\cdots X_{m_1}^{-v_{m_1}}g \mid g\in V_{\Delta_2}\right\},$$
and the codewords in $C^P_{\Delta_2^{'}}$ are of the form
$$\ev_P(X_1^{-v_1}\cdots X_{m_1}^{-v_{m_1}}g)=\ev_P(X_1^{-v_1}\cdots X_{m_1}^{-v_{m_1}})*\ev_P(g),$$
where $g\in V_{\Delta_2}$. When $0\notin P_j$ for all $1\leq j \leq m$, we have just proved that $C^P_{\Delta_2^{'}}$ and $C^P_{\Delta_2}$ are isometric codes. The same reasoning proves that $C^P_{\Delta_1^{'}}$ and $C^P_{\Delta_1}$ are isometric. Thus $C^P_{\Delta_1}$ and $C^P_{\Delta_2}$ are isometric and this happens even when the $v_j$ are always negative or positive or when some coordinates $v_j$ are $0$.

When $0\in P_j$ for some index $1\leq j\leq m$, $C^P_{\Delta_1}$ and $C^P_{\Delta_2}$ need not be isometric which explains why we speak of pseudoisometric codes.
\end{remark}

Length, dimension and a bound for the minimum distance of an MCC, $C^P_\Delta$, are provided in the forthcoming Proposition \ref{lendim} and Corollary \ref{pdist}. Let us state Proposition \ref{lendim} whose proof is straightforward.

\begin{prop}\label{lendim}
Keep the above notation. The length $n$ and dimension $k$ of an MCC, $C^P_\Delta$, are $n=\prod_{j=1}^m n_j$ and $k=\#\Delta$.
\end{prop}

\begin{mydef}\label{dist}
The \textit{distance} of an exponent $\mathbf{e}\in E$ is defined to be $\dis(\mathbf{e}):=\prod_{j=1}^m (n_j-e_j)$.
\end{mydef}

The codes $C^P_\Delta$ admit the following bound on the minimum distance, known as footprint bound \cite{G2008,GGHR2017}.

\begin{prop}\label{footprint}
Let $C_\Delta^P$ be an MCC and let $\mathbf{c}=\ev_P(f)\in C_\Delta^P$ be a codeword, $f\in \mathcal{R}$. Denote by $\w(\mathbf{c})$ the Hamming weight of $\mathbf{c}$, fix a monomial ordering on $(\mathbb{Z}_{\geq 0})^m$ and let $X^\mathbf{e}$ be the leading monomial of $f$. Then, $\w(\mathbf{c})\geq \dis(\mathbf{e})$.
\end{prop}

\begin{col}\label{pdist}
Let $C_\Delta^P$ be an MCC and let $d$ be its minimum distance. Define $d_0=d_0\left(C_\Delta^P\right):=\min\{\dis(\mathbf{e}) \mid \mathbf{e}\in \Delta\}$. Then, $d\geq d_0$.
\end{col}

\begin{remark}\label{charact}
With the above notation, given $\emptyset\neq\Delta\subseteq E$, define $M_\Delta:=\{X^\mathbf{e} \mid \mathbf{e}\in \Delta\}$. According to \cite[Definition 3.1]{CLMS2021}, a code $C_\Delta^P$ is named \textit{decreasing monomial-Cartesian} whenever
\begin{equation}\label{decreasing}
X^\mathbf{e} \in M_\Delta \text{ implies } X^\mathbf{e'}\in M_\Delta \text{ for all } \mathbf{e'}\in E \text{ such that } X^\mathbf{e'} \text{ divides } X^\mathbf{e}.
\end{equation}
Moreover, by \cite[Theorem 3.9]{CLMS2021}, the values $d$ and $d_0$ of decreasing MCCs coincide.
\end{remark}

\begin{mydef}\label{decr}
A set $\Delta\subseteq E$ that satisfies (\ref{decreasing}) is called \textit{decreasing}.
\end{mydef}

Next proposition and its proof show how to regard MCCs as LRCs. To do it, we need to introduce some definitions. For each $1\leq j \leq m$, define the support of $V_\Delta$ at $X_j$ as
$$\supp_{X_j}(V_\Delta):=\left\{e_j\in\{0,1,\dots,n_j-1\} \mid \textrm{there exists a monomial } X_1^{e_1}\cdots X_j^{e_j}\cdots X_m^{e_m} \textrm{ in } V_\Delta\right\},$$
and set $\mathcal{K}_j:=\#\supp_{X_j}(V_\Delta)$ and $k_j:=\max\left(\supp_{X_j}(V_\Delta)\right)$. Now, and as the beginning of this section, set $P_j=\{\alpha_1,\dots,\alpha_{n_j}\}\subseteq \mathbb{F}_q$, $I_j$ the ideal of $\mathbb{F}_q[X_j]$ generated by $f_j=\prod_{\beta\in P_j}(X_j-\beta)$ and
$$\ev_{P_j}\colon \mathcal{R}_j:=\faktor{\mathbb{F}_q[X_j]}{I_j} \to \mathbb{F}_q^{n_j}$$
given by
$$\ev_{P_j}(f)=\left(f(\alpha_1),\dots,f(\alpha_{n_j})\right).$$
Finally define $V^j_\Delta:=\langle X_j^e \mid e\in \supp_{X_j}(V_\Delta) \rangle _{\mathbb{F}_q} \subseteq \mathcal{R}_j$.

\begin{prop}\label{method}
Let $C_\Delta^P$ be an MCC. Then, for each $1\leq l \leq m$ such that $\mathcal{K}_l<n_l$, $C_\Delta^P$ is an LRC with locality $(\geq \mathcal{K}_l,\leq n_l-\mathcal{K}_l+1)$. In addition, if $\ev_{P_l}\left(V^l_\Delta\right)$ is an MDS code, then the locality is $(\mathcal{K}_l, n_l-\mathcal{K}_l+1)$.
\end{prop}

\begin{proof}
Let $\mathbf{c}=(c_1,\dots,c_n)=\ev_P(f)\in C_\Delta^P$ be a codeword whose $i$th coordinate $c_i$ we desire to recover. We know that $\supp(f)\subseteq \Delta$ and thus $\deg_{X_j}(f)\leq k_j$ for all $j=1,\dots,m$. Choose a variable $X_l$ (we will interpolate with respect to it), write $c_i=f(\boldsymbol{\alpha}_i)=f({\alpha_i}_1,\dots,{\alpha_i}_m)$ and consider the following subset of $P$:
\begin{multline*}
    \overline R_P=\left\{\boldsymbol{\alpha}_t\in P \mid {\alpha_t}_j={\alpha_i}_j \textrm{ for all } j\in\{1,\dots,m\}\backslash \{l\}\right\}\\
    =\left\{({\alpha_i}_1,\dots,{\alpha_i}_{l-1},x,{\alpha_i}_{l+1},\dots,{\alpha_i}_m) \mid x\in P_l\right\},
\end{multline*}
whose cardinality is $\#\overline R_P=n_l$. A polynomial in $V_\Delta$ can be expressed as
\begin{multline*}
f(X_1,\dots,X_m)=\sum_{(e_1,\dots,e_m)\in\Delta}f_{e_1,\dots,e_m}X_1^{e_1}\cdots X_m^{e_m}\\
=\sum_{h=0}^{k_l}f_h(X_1,\dots,X_{l-1},X_{l+1},\dots,X_m)X_l^h \in \mathbb{F}_q[X_1,\dots,X_{l-1},X_{l+1},\dots,X_m][X_l].
\end{multline*}
Replacing each $X_j$, $j\neq l$, by ${\alpha_i}_j$, we get a polynomial in $X_l$, $g(X_l)$, with constant coefficients, of degree at most $k_l$. So we can interpolate $g$ by using $k_l+1$ points in $\overline R_P$ (since $k_l\leq n_l-1$) to obtain those coefficients. However, we have $k_l+1-\mathcal{K}_l$ conditions
$$f_h({\alpha_i}_1,\dots,{\alpha_i}_{l-1},{\alpha_i}_{l+1},\dots,{\alpha_i}_m)=0$$
$h\notin \supp_{X_l}(V_\Delta)$, and then we only need $\mathcal{K}_l$ points in $\overline R_P$ to obtain the coefficients of $g$. Recall that $\boldsymbol{\alpha}_i\in\overline R_P$ implies that $\mathcal{K}_l<n_l$. Then, we can recover $c_i$ by evaluating $g$. Let
$$\overline R=\{t\in\{1,\dots,n\} \mid \boldsymbol{\alpha}_t \in \overline R_P\}.$$
The set $\overline R$ is an $(r,\delta)$-recovery set for $i$ with $r:=n_l-\dis(C[\overline R])+1$ and $\delta:=\dis(C[\overline R])$ since $i\in \overline R$, $\#\overline R=n_l=r+\delta-1$ and $\dis(C[\overline R])=\delta$. The Singleton bound implies that $\delta=\dis(C[\overline R])\leq n_l-\dim(C[\overline R])+1=n_l-\mathcal{K}_l+1$ and therefore $r\geq \mathcal{K}_l$. Hence, $C_\Delta^P$ is an LRC with locality $(\geq \mathcal{K}_l, \leq n_l-\mathcal{K}_l+1)$.

Our last statement follows from the fact that when $C[\overline R]=\ev_{P_l}\left(V^l_\Delta\right)$ is an MDS code, then the locality is $(\mathcal{K}_l,n_l-\mathcal{K}_l+1)$.
\end{proof}

\begin{remark}
With the above notation and when $\supp_{X_l}(V_\Delta)=\{0,1,\dots,k_l\}$, it holds that $\ev_{P_l}\left(V^l_\Delta\right)$ is a Reed-Solomon code (and thus an MDS code), and then the locality of $C_\Delta^P$ is $(\mathcal{K}_l, n_l-\mathcal{K}_l+1)$.
\end{remark}

\begin{remark}
Let $C_\Delta^P$ be an MCC with parameters $[n,k,d]_q$ and locality $(r,\delta)$. Then by Proposition \ref{pformula} and Corollary \ref{pdist}, the following inequalities
\begin{equation}\label{formula2}
    k+d_0+\left(\ceil*{\frac{k}{r}}-1\right)(\delta-1)\leq k+d+\left(\ceil*{\frac{k}{r}}-1\right)(\delta-1)\leq n+1
\end{equation}
hold.
\end{remark}

Let $C_\Delta^P$ be an MCC with parameters $[n,k,d]_q$ and locality $(r,\delta)$. We define its defect (with respect to $d_0$) as the value $D$:
$$D=D\left(C_\Delta^P\right):=n+1-k-d_0-\left(\ceil*{\frac{k}{r}}-1\right)(\delta-1)\geq 0.$$

\begin{mydef}\label{d0op}
The code $C_\Delta^P$ is called \textit{$d_0$-optimal} whenever $D$ vanishes. That is, $C_\Delta^P$ is optimal and $d=d_0$.
\end{mydef}

\begin{remarks}\label{facts}
The next facts will be useful:
\begin{enumerate}
    \item The locality $(r,\delta)$ provided in Proposition \ref{method} depends on the variable $X_l$ we choose to interpolate, which allows us to make the best choice of $X_l$.
    
    \item A $d_0$-optimal code is always optimal but a code that is not $d_0$-optimal may be optimal.
\end{enumerate}
\end{remarks}

\section{Optimal monomial-Cartesian codes}\label{work}

In this section we obtain optimal decreasing MCCs. We start with the bivariate case.

\subsection{The case $m=2$}\label{subm2}
For simplicity let us denote $X_1$ by $X$ and $X_2$ by $Y$. We represent $E$ as a grid where the coordinates $(i,j)$ correspond to an exponent $\mathbf{e}$ labelled with their distance (Definition \ref{dist}). Figure \ref{fig:grid} shows the grid representation of $E$ in the case when $n_1=10$ and $n_2=9$.

\begin{figure}[h]
    \centering
    \begin{tikzpicture}[y=0.7cm, x=0.7cm,font=\normalsize]
    \draw (0,0) -- (9,0);
    \draw (0,0) -- (0,8);
    \draw (0,8) -- (9,8);
    \draw (9,0) -- (9,8);
    \draw (1,0) -- (1,8);
    \draw (2,0) -- (2,8);
    \draw (3,0) -- (3,8);
    \draw (4,0) -- (4,8);
    \draw (5,0) -- (5,8);
    \draw (6,0) -- (6,8);
    \draw (7,0) -- (7,8);
    \draw (8,0) -- (8,8);
    \draw (0,1) -- (9,1);
    \draw (0,2) -- (9,2);
    \draw (0,3) -- (9,3);
    \draw (0,4) -- (9,4);
    \draw (0,5) -- (9,5);
    \draw (0,6) -- (9,6);
    \draw (0,7) -- (9,7);

    \filldraw[fill=black!40,draw=black!80] (0,0) circle (1pt)    node[anchor=south] {\scriptsize$90$};
    \filldraw[fill=black!40,draw=black!80] (1,0) circle (1pt)    node[anchor=south] {\scriptsize$81$};
    \filldraw[fill=black!40,draw=black!80] (2,0) circle (1pt)    node[anchor=south] {\scriptsize$72$};
    \filldraw[fill=black!40,draw=black!80] (3,0) circle (1pt)    node[anchor=south] {\scriptsize$63$};
    \filldraw[fill=black!40,draw=black!80] (4,0) circle (1pt)    node[anchor=south] {\scriptsize$54$};
    \filldraw[fill=black!40,draw=black!80] (5,0) circle (1pt)    node[anchor=south] {\scriptsize$45$};
    \filldraw[fill=black!40,draw=black!80] (6,0) circle (1pt)    node[anchor=south] {\scriptsize$36$};
    \filldraw[fill=black!40,draw=black!80] (7,0) circle (1pt)    node[anchor=south] {\scriptsize$27$};
    \filldraw[fill=black!40,draw=black!80] (8,0) circle (1pt)    node[anchor=south] {\scriptsize$18$};
    \filldraw[fill=black!40,draw=black!80] (9,0) circle (1pt)    node[anchor=south] {\scriptsize$9$};
    \filldraw[fill=black!40,draw=black!80] (0,1) circle (1pt)    node[anchor=south] {\scriptsize$80$};
    \filldraw[fill=black!40,draw=black!80] (1,1) circle (1pt)    node[anchor=south] {\scriptsize$72$};
    \filldraw[fill=black!40,draw=black!80] (2,1) circle (1pt)    node[anchor=south] {\scriptsize$64$};
    \filldraw[fill=black!40,draw=black!80] (3,1) circle (1pt)    node[anchor=south] {\scriptsize$56$};
    \filldraw[fill=black!40,draw=black!80] (4,1) circle (1pt)    node[anchor=south] {\scriptsize$48$};
    \filldraw[fill=black!40,draw=black!80] (5,1) circle (1pt)    node[anchor=south] {\scriptsize$40$};
    \filldraw[fill=black!40,draw=black!80] (6,1) circle (1pt)    node[anchor=south] {\scriptsize$32$};
    \filldraw[fill=black!40,draw=black!80] (7,1) circle (1pt)    node[anchor=south] {\scriptsize$24$};
    \filldraw[fill=black!40,draw=black!80] (8,1) circle (1pt)    node[anchor=south] {\scriptsize$16$};
    \filldraw[fill=black!40,draw=black!80] (9,1) circle (1pt)    node[anchor=south] {\scriptsize$8$};
    \filldraw[fill=black!40,draw=black!80] (0,2) circle (1pt)    node[anchor=south] {\scriptsize$70$};
    \filldraw[fill=black!40,draw=black!80] (1,2) circle (1pt)    node[anchor=south] {\scriptsize$63$};
    \filldraw[fill=black!40,draw=black!80] (2,2) circle (1pt)    node[anchor=south] {\scriptsize$56$};
    \filldraw[fill=black!40,draw=black!80] (3,2) circle (1pt)    node[anchor=south] {\scriptsize$49$};
    \filldraw[fill=black!40,draw=black!80] (4,2) circle (1pt)    node[anchor=south] {\scriptsize$42$};
    \filldraw[fill=black!40,draw=black!80] (5,2) circle (1pt)    node[anchor=south] {\scriptsize$35$};
    \filldraw[fill=black!40,draw=black!80] (6,2) circle (1pt)    node[anchor=south] {\scriptsize$28$};
    \filldraw[fill=black!40,draw=black!80] (7,2) circle (1pt)    node[anchor=south] {\scriptsize$21$};
    \filldraw[fill=black!40,draw=black!80] (8,2) circle (1pt)    node[anchor=south] {\scriptsize$14$};
    \filldraw[fill=black!40,draw=black!80] (9,2) circle (1pt)    node[anchor=south] {\scriptsize$7$};
    \filldraw[fill=black!40,draw=black!80] (0,3) circle (1pt)    node[anchor=south] {\scriptsize$60$};
    \filldraw[fill=black!40,draw=black!80] (1,3) circle (1pt)    node[anchor=south] {\scriptsize$54$};
    \filldraw[fill=black!40,draw=black!80] (2,3) circle (1pt)    node[anchor=south] {\scriptsize$48$};
    \filldraw[fill=black!40,draw=black!80] (3,3) circle (1pt)    node[anchor=south] {\scriptsize$42$};
    \filldraw[fill=black!40,draw=black!80] (4,3) circle (1pt)    node[anchor=south] {\scriptsize$36$};
    \filldraw[fill=black!40,draw=black!80] (5,3) circle (1pt)    node[anchor=south] {\scriptsize$30$};
    \filldraw[fill=black!40,draw=black!80] (6,3) circle (1pt)    node[anchor=south] {\scriptsize$24$};
    \filldraw[fill=black!40,draw=black!80] (7,3) circle (1pt)    node[anchor=south] {\scriptsize$18$};
    \filldraw[fill=black!40,draw=black!80] (8,3) circle (1pt)    node[anchor=south] {\scriptsize$12$};
    \filldraw[fill=black!40,draw=black!80] (9,3) circle (1pt)    node[anchor=south] {\scriptsize$6$};
    \filldraw[fill=black!40,draw=black!80] (0,4) circle (1pt)    node[anchor=south] {\scriptsize$50$};
    \filldraw[fill=black!40,draw=black!80] (1,4) circle (1pt)    node[anchor=south] {\scriptsize$45$};
    \filldraw[fill=black!40,draw=black!80] (2,4) circle (1pt)    node[anchor=south] {\scriptsize$40$};
    \filldraw[fill=black!40,draw=black!80] (3,4) circle (1pt)    node[anchor=south] {\scriptsize$35$};
    \filldraw[fill=black!40,draw=black!80] (4,4) circle (1pt)    node[anchor=south] {\scriptsize$30$};
    \filldraw[fill=black!40,draw=black!80] (5,4) circle (1pt)    node[anchor=south] {\scriptsize$25$};
    \filldraw[fill=black!40,draw=black!80] (6,4) circle (1pt)    node[anchor=south] {\scriptsize$20$};
    \filldraw[fill=black!40,draw=black!80] (7,4) circle (1pt)    node[anchor=south] {\scriptsize$15$};
    \filldraw[fill=black!40,draw=black!80] (8,4) circle (1pt)    node[anchor=south] {\scriptsize$10$};
    \filldraw[fill=black!40,draw=black!80] (9,4) circle (1pt)    node[anchor=south] {\scriptsize$5$};
    \filldraw[fill=black!40,draw=black!80] (0,5) circle (1pt)    node[anchor=south] {\scriptsize$40$};
    \filldraw[fill=black!40,draw=black!80] (1,5) circle (1pt)    node[anchor=south] {\scriptsize$36$};
    \filldraw[fill=black!40,draw=black!80] (2,5) circle (1pt)    node[anchor=south] {\scriptsize$32$};
    \filldraw[fill=black!40,draw=black!80] (3,5) circle (1pt)    node[anchor=south] {\scriptsize$28$};
    \filldraw[fill=black!40,draw=black!80] (4,5) circle (1pt)    node[anchor=south] {\scriptsize$24$};
    \filldraw[fill=black!40,draw=black!80] (5,5) circle (1pt)    node[anchor=south] {\scriptsize$20$};
    \filldraw[fill=black!40,draw=black!80] (6,5) circle (1pt)    node[anchor=south] {\scriptsize$16$};
    \filldraw[fill=black!40,draw=black!80] (7,5) circle (1pt)    node[anchor=south] {\scriptsize$12$};
    \filldraw[fill=black!40,draw=black!80] (8,5) circle (1pt)    node[anchor=south] {\scriptsize$8$};
    \filldraw[fill=black!40,draw=black!80] (9,5) circle (1pt)    node[anchor=south] {\scriptsize$4$};
    \filldraw[fill=black!40,draw=black!80] (0,6) circle (1pt)    node[anchor=south] {\scriptsize$30$};
    \filldraw[fill=black!40,draw=black!80] (1,6) circle (1pt)    node[anchor=south] {\scriptsize$27$};
    \filldraw[fill=black!40,draw=black!80] (2,6) circle (1pt)    node[anchor=south] {\scriptsize$24$};
    \filldraw[fill=black!40,draw=black!80] (3,6) circle (1pt)    node[anchor=south] {\scriptsize$21$};
    \filldraw[fill=black!40,draw=black!80] (4,6) circle (1pt)    node[anchor=south] {\scriptsize$18$};
    \filldraw[fill=black!40,draw=black!80] (5,6) circle (1pt)    node[anchor=south] {\scriptsize$15$};
    \filldraw[fill=black!40,draw=black!80] (6,6) circle (1pt)    node[anchor=south] {\scriptsize$12$};
    \filldraw[fill=black!40,draw=black!80] (7,6) circle (1pt)    node[anchor=south] {\scriptsize$9$};
    \filldraw[fill=black!40,draw=black!80] (8,6) circle (1pt)    node[anchor=south] {\scriptsize$6$};
    \filldraw[fill=black!40,draw=black!80] (9,6) circle (1pt)    node[anchor=south] {\scriptsize$3$};
    \filldraw[fill=black!40,draw=black!80] (0,7) circle (1pt)    node[anchor=south] {\scriptsize$20$};
    \filldraw[fill=black!40,draw=black!80] (1,7) circle (1pt)    node[anchor=south] {\scriptsize$18$};
    \filldraw[fill=black!40,draw=black!80] (2,7) circle (1pt)    node[anchor=south] {\scriptsize$16$};
    \filldraw[fill=black!40,draw=black!80] (3,7) circle (1pt)    node[anchor=south] {\scriptsize$14$};
    \filldraw[fill=black!40,draw=black!80] (4,7) circle (1pt)    node[anchor=south] {\scriptsize$12$};
    \filldraw[fill=black!40,draw=black!80] (5,7) circle (1pt)    node[anchor=south] {\scriptsize$10$};
    \filldraw[fill=black!40,draw=black!80] (6,7) circle (1pt)    node[anchor=south] {\scriptsize$8$};
    \filldraw[fill=black!40,draw=black!80] (7,7) circle (1pt)    node[anchor=south] {\scriptsize$6$};
    \filldraw[fill=black!40,draw=black!80] (8,7) circle (1pt)    node[anchor=south] {\scriptsize$4$};
    \filldraw[fill=black!40,draw=black!80] (9,7) circle (1pt)    node[anchor=south] {\scriptsize$2$};
    \filldraw[fill=black!40,draw=black!80] (0,8) circle (1pt)    node[anchor=south] {\scriptsize$10$};
    \filldraw[fill=black!40,draw=black!80] (1,8) circle (1pt)    node[anchor=south] {\scriptsize$9$};
    \filldraw[fill=black!40,draw=black!80] (2,8) circle (1pt)    node[anchor=south] {\scriptsize$8$};
    \filldraw[fill=black!40,draw=black!80] (3,8) circle (1pt)    node[anchor=south] {\scriptsize$7$};
    \filldraw[fill=black!40,draw=black!80] (4,8) circle (1pt)    node[anchor=south] {\scriptsize$6$};
    \filldraw[fill=black!40,draw=black!80] (5,8) circle (1pt)    node[anchor=south] {\scriptsize$5$};
    \filldraw[fill=black!40,draw=black!80] (6,8) circle (1pt)    node[anchor=south] {\scriptsize$4$};
    \filldraw[fill=black!40,draw=black!80] (7,8) circle (1pt)    node[anchor=south] {\scriptsize$3$};
    \filldraw[fill=black!40,draw=black!80] (8,8) circle (1pt)    node[anchor=south] {\scriptsize$2$};
    \filldraw[fill=black!40,draw=black!80] (9,8) circle (1pt)    node[anchor=south] {\scriptsize$1$};

    \node [below] at (0,0) {\scriptsize$0$};
    \node [below] at (1,0) {\scriptsize$1$};
    \node [below] at (2,0) {\scriptsize$2$};
    \node [below] at (3,0) {\scriptsize$3$};
    \node [below] at (4,0) {\scriptsize$4$};
    \node [below] at (5,0) {\scriptsize$5$};
    \node [below] at (6,0) {\scriptsize$6$};
    \node [below] at (7,0) {\scriptsize$7$};
    \node [below] at (8,0) {\scriptsize$8$};
    \node [below] at (9,0) {\scriptsize$9$};
    \node [left] at (0,0) {\scriptsize$0$};
    \node [left] at (0,1) {\scriptsize$1$};
    \node [left] at (0,2) {\scriptsize$2$};
    \node [left] at (0,3) {\scriptsize$3$};
    \node [left] at (0,4) {\scriptsize$4$};
    \node [left] at (0,5) {\scriptsize$5$};
    \node [left] at (0,6) {\scriptsize$6$};
    \node [left] at (0,7) {\scriptsize$7$};
    \node [left] at (0,8) {\scriptsize$8$};
    
    \end{tikzpicture}
    \caption{Grid representation of $E$, where $n_1=10$ and $n_2=9$}
    \label{fig:grid}
\end{figure}
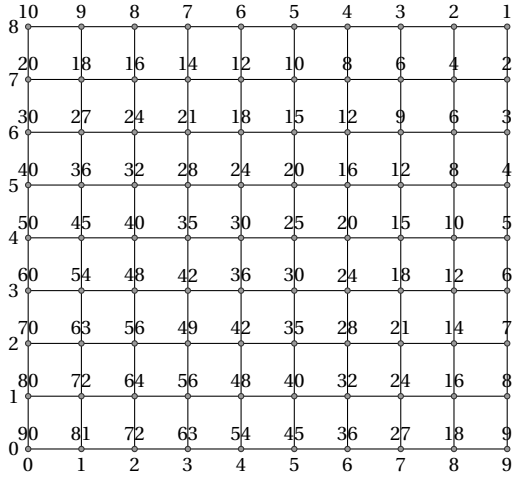

We look for decreasing sets $\Delta\subseteq E$ such that the code $C_\Delta^P$ is optimal, that is, its parameters satisfy
$$k+d_0+\left(\ceil*{\frac{k}{r}}-1\right)(\delta-1)=n+1.$$
Note that, by Remark \ref{charact}, $d=d_0$.

From now on, we use shaded regions to represent sets formed by the points in $E$ inside that region. By rectangle we will always refer to a subset of $E$ whose representation as shaded set is a rectangle. The first result in this subsection shows when codes $C_\Delta^P$, where $\Delta$ is decreasing and has the shape of a rectangle, are optimal.

\begin{figure}[h]
    \centering
    \begin{tikzpicture}[y=0.7cm, x=0.7cm,font=\normalsize]

    \filldraw[fill=gray!30] (0,0) rectangle (2,2);
    \draw (0,2) -- (0,4);
    \draw (2,0) -- (4,0);
    
    \filldraw[fill=black!40,draw=black!80] (2,2) circle (1pt)    node[anchor=south] {\scriptsize$(n_1-i)(n_2-j)$};
    
    \node [below] at (0,0) {\scriptsize$0$};
    \node [below] at (1,0) {$\dots$};
    \node [below] at (2,0) {\scriptsize$i$};
    \node [below] at (3,0) {$\dots$};
    \node [below] at (4,0) {\scriptsize$n_1-1$};
    \node [left] at (0,0) {\scriptsize$0$};
    \node [left] at (0,1) {$\vdots$};
    \node [left] at (0,2) {\scriptsize$j$};
    \node [left] at (0,3) {$\vdots$};
    \node [left] at (0,4) {\scriptsize$n_2-1$};
    
    \end{tikzpicture}
    \caption{Sets $\Delta_{i,j}$ in Proposition \ref{rect}}
    \label{fig:rect}
\end{figure}
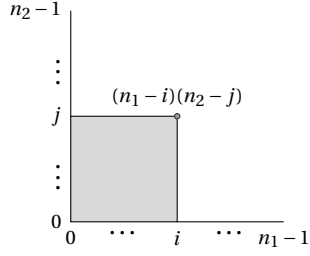

\begin{prop}\label{rect}
Keep the above notation, where $q$ is a prime power, $m=2$ and $n_1$, $n_2\geq 2$ are the cardinalities of $P_1$ and $P_2$. Consider the sets
$$\Delta=\Delta_{i,j}:=\left\{(e_1,e_2) \mid 0\leq e_1 \leq i \textrm{, } 0\leq e_2\leq j\right\}\subseteq E=\{0,\dots,n_1-1\}\times\{0,\dots,n_2-1\}$$
(see Figure \ref{fig:rect}). Then, the MCC, $C_\Delta^P$, defined by a set $\Delta$ as above is an optimal $(r,\delta)$-LRC if and only if one of the following conditions hold:
\begin{itemize}
    \item $i=0$ and $0\leq j\leq n_2-1$, in which case $(r,\delta)=(1,n_1)$.
    \item $1\leq i\leq n_1-2$ and $j=n_2-1$, in which case $(r,\delta)=(i+1,n_1-i)$.
    \item $0\leq i\leq n_1-1$ and $j=0$, in which case $(r,\delta)=(1,n_2)$.
    \item $i=n_1-1$ and $1\leq j\leq n_2-2$, in which case $(r,\delta)=(j+1,n_2-j)$.
\end{itemize}
Sets $\Delta$ as above are denoted by $\Delta_{i,j}^1$.
\end{prop}

\begin{proof}
Clearly, $k=(i+1)(j+1)$ and $d_0=(n_1-i)(n_2-j)$. By interpolating with respect to $X$, $r=i+1$ and $\delta-1=n_1-i-1$. Then,
\begin{equation*}
\begin{aligned}
k+d_0+\left(\ceil*{\frac{k}{r}}-1\right)(\delta-1)&=(i+1)(j+1)+(n_1-i)(n_2-j)+\left(\ceil*{\frac{(i+1)(j+1)}{i+1}}-1\right)(n_1-i-1)\\
&=n_1n_2+1+i(j+1-n_2),
\end{aligned}
\end{equation*}
and the code is optimal if and only if $i=0$ or $j=n_2-1$. Note that when $j=n_2-1$ and $i=n_1-1$ one does not get an LRC.

The remaining LRCs are obtained by interpolating with respect to $Y$, so that $r=j+1$ and $\delta-1=n_2-j-1$.
\end{proof}

In the sequel, we will perform the procedure of considering a subset $\Delta\subseteq E$ and adding or removing elements to obtain a new subset $\Delta^*\subseteq E$. The expression \textit{gaining} (or \textit{losing}) $x$ \textit{units in} a parameter refers to the fact that the resulting code $C_{\Delta^*}^P$ has a larger (or smaller) value for that parameter in a quantity of $x$ units.

The sets $\Delta^*$ obtained by removing the least distance point on the $n_2-1$-th row (or $n_1-1$-th column) of a rectangle $\Delta_{i,j}^1$ with $j=n_2-1$ and $i\geq 1$ (or $i=n_1-1$ and $j\geq 1$) also provide optimal codes since the left-hand side (LHS) of (\ref{formula2}) remains the same. Indeed, when removing that point we lose one unit in dimension but we gain one unit in the bound for the minimum distance and $r$, $\delta$ and $\ceil*{\frac{k}{r}}$ do not change. The following result generalizes this situation.

\begin{figure}[h]
    \centering
    \begin{subfigure}[b]{0.45\textwidth}
        \centering
        \begin{tikzpicture}[y=0.7cm, x=0.7cm,font=\normalsize]

        \filldraw[gray!30] (0,0) rectangle (4,5);
        \filldraw[gray!30] (0,5) rectangle (2,6);
        \draw (0,0) -- (6,0);
        \draw (0,0) -- (0,6);
        \draw (0,6) -- (2,6);
        \draw (2,6) -- (2,5);
        \draw (2,5) -- (4,5);
        \draw (4,0) -- (4,5);

        \filldraw[fill=black!40,draw=black!80] (2,6) circle (1pt)    node[anchor=south] {\scriptsize$n_1-s$};
        \filldraw[fill=black!40,draw=black!80] (4,5) circle (1pt)    node[anchor=south] {\scriptsize$2(n_1-i)$};
        
        \node [below] at (0,0) {\scriptsize$0$};
        \node [below] at (1,0) {$\dots$};
        \node [below] at (2,0) {\scriptsize$s$};
        \node [below] at (3,0) {$\dots$};
        \node [below] at (4,0) {\scriptsize$i$};
        \node [below] at (5,0) {$\dots$};
        \node [below] at (6,0) {\scriptsize$n_1-1$};
        \node [left] at (0,0) {\scriptsize$0$};
        \node [left] at (0,3) {$\vdots$};
        \node [left] at (0,5) {\scriptsize$n_2-2$};
        \node [left] at (0,6) {\scriptsize$n_2-1$};
        
        \end{tikzpicture}
        \caption{Sets $\Delta_{i,s}^2$}
    \end{subfigure}
    \hfill
    \begin{subfigure}[b]{0.45\textwidth}
        \centering
        \begin{tikzpicture}[y=0.7cm, x=0.7cm,font=\normalsize]

        \filldraw[gray!30] (0,0) rectangle (5,4);
        \filldraw[gray!30] (5,0) rectangle (6,2);
        \draw (0,0) -- (0,6);
        \draw (0,0) -- (6,0);
        \draw (6,0) -- (6,2);
        \draw (6,2) -- (5,2);
        \draw (5,2) -- (5,4);
        \draw (0,4) -- (5,4);

        \filldraw[fill=black!40,draw=black!80] (6,2) circle (1pt)    node[anchor=south] {\scriptsize$n_2-s$};
        \filldraw[fill=black!40,draw=black!80] (5,4) circle (1pt)    node[anchor=south] {\scriptsize$2(n_2-j)$};
        
        \node [left] at (0,0) {\scriptsize$0$};
        \node [left] at (0,1) {$\vdots$};
        \node [left] at (0,2) {\scriptsize$s$};
        \node [left] at (0,3) {$\vdots$};
        \node [left] at (0,4) {\scriptsize$j$};
        \node [left] at (0,5) {$\vdots$};
        \node [left] at (0,6) {\scriptsize$n_2-1$};
        \node [below] at (0,0) {\scriptsize$0$};
        \node [below] at (3,0) {$\dots$};
        \node [below] at (5,0) {\scriptsize$n_1-2,$};
        \node [below] at (6,0) {\scriptsize$n_1-1$};

        \end{tikzpicture}
        \caption{Sets $\Delta_{j,s}^{2,\sigma}$}
    \end{subfigure}
    \caption{Sets $\Delta_{i,s}^2$ and $\Delta_{j,s}^{2,\sigma}$ in Proposition \ref{rectelim}}
    \label{fig:rectelim}
\end{figure}

\begin{prop}\label{rectelim}
With notation as in Proposition \ref{rect}, consider the subsets of $E$
$$\Delta=\Delta_{i,s}^2:=\left\{(e_1,e_2) \mid 0\leq e_1 \leq i \textrm{, } 0\leq e_2 \leq n_2-2\right\}\cup \left\{(e_1,n_2-1) \mid 0\leq e_1\leq s\right\},$$
where $\max\left\{0,2i-n_1\right\}\leq s<i\leq n_1-2$ (see Figure \ref{fig:rectelim} (1)).

Then, the MCCs, $C_\Delta^P$, are optimal $(r,\delta)=(i+1,n_1-i)$-LRCs.

Analogously, the MCCs, $C_\Delta^P$, where
$$\Delta=\Delta_{j,s}^{2,\sigma}:=\left\{(e_1,e_2) \mid 0\leq e_1 \leq n_1-2 \textrm{, } 0\leq e_2 \leq j\right\}\cup \left\{(n_1-1,e_2) \mid 0\leq e_2\leq s\right\}\subseteq E,$$
$\max\left\{0,2j-n_2\right\}\leq s<j\leq n_2-2$ (see Figure \ref{fig:rectelim} (2)) are optimal $(r,\delta)=(j+1,n_2-j)$-LRCs.
\end{prop}

\begin{proof}
Let us see a proof for the case $\Delta=\Delta_{i,s}^2$. $\Delta$ is obtained by removing the ($i-s$) least distance points of $\Delta_{i,n_2-1}^1$ on the $n_2-1$-th row with $0\leq s<i$ as long as the distance
$$\dis(s,n_2-1)\leq\dis(i,n_2-2).$$
In fact, this last inequality is equivalent to $n_1-s\leq 2(n_1-i)$ and to $s\geq 2i-n_1$. Interpolating with respect to $X$, the parameters of the code $C_\Delta^P$ are $k=(i+1)(n_2-1)+s+1$, $d_0=n_1-s$, $r=i+1$ and $\delta-1=n_1-i-1$, and therefore
\begin{equation*}
\begin{aligned}
k+d_0+\left(\ceil*{\frac{k}{r}}-1\right)(\delta-1)&=(i+1)(n_2-1)+s+1+n_1-s\\
&+\left(\ceil*{\frac{(i+1)(n_2-1)+s+1}{i+1}}-1\right)(n_1-i-1)\\
&=n_1n_2+1.
\end{aligned}
\end{equation*}
The case $\Delta=\Delta_{j,s}^{2,\sigma}$ can be proved analogously. It suffices to consider the symmetric situation, interpolate with respect to $Y$ and replace $i$ by $j$ and $n_1$ by $n_2$.
\end{proof}

The following result completes our family of decreasing sets $\Delta$, that correspond to MCCs, where $m=2$, giving rise to optimal $(r,\delta)$-LRCs.

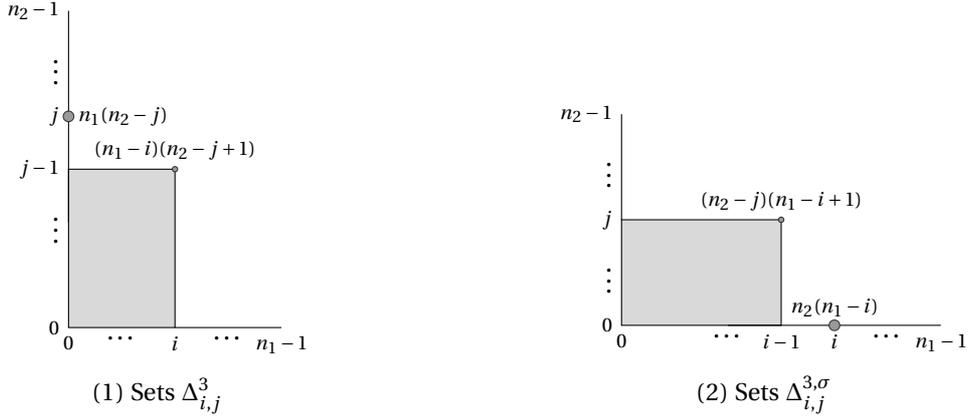
\begin{figure}[h]
    \centering
    \begin{subfigure}[b]{0.45\textwidth}
        \centering
        \begin{tikzpicture}[y=0.7cm, x=0.7cm,font=\normalsize]

        \filldraw[fill=gray!30] (0,0) rectangle (2,3);
        \draw (0,2) -- (0,6);
        \draw (2,0) -- (4,0);
        
        \filldraw[fill=black!40,draw=black!80] (0,4) circle (2pt)    node[anchor=west] {\scriptsize$n_1(n_2-j)$};
        \filldraw[fill=black!40,draw=black!80] (2,3) circle (1pt)    node[anchor=south] {\scriptsize$(n_1-i)(n_2-j+1)$};
        
        \node [below] at (0,0) {\scriptsize$0$};
        \node [below] at (1,0) {$\dots$};
        \node [below] at (2,0) {\scriptsize$i$};
        \node [below] at (3,0) {$\dots$};
        \node [below] at (4,0) {\scriptsize$n_1-1$};
        \node [left] at (0,0) {\scriptsize$0$};
        \node [left] at (0,2) {$\vdots$};
        \node [left] at (0,3) {\scriptsize$j-1$};
        \node [left] at (0,4) {\scriptsize$j$};
        \node [left] at (0,5) {$\vdots$};
        \node [left] at (0,6) {\scriptsize$n_2-1$};
        
        \end{tikzpicture}
        \caption{Sets $\Delta_{i,j}^3$}
    \end{subfigure}
    \hfill
    \begin{subfigure}[b]{0.45\textwidth}
        \centering
        \begin{tikzpicture}[y=0.7cm, x=0.7cm,font=\normalsize]

        \filldraw[fill=gray!30] (0,0) rectangle (3,2);
        \draw (2,0) -- (6,0);
        \draw (0,2) -- (0,4);
        
        \filldraw[fill=black!40,draw=black!80] (4,0) circle (2pt)    node[anchor=south] {\scriptsize$n_2(n_1-i)$};
        \filldraw[fill=black!40,draw=black!80] (3,2) circle (1pt)    node[anchor=south] {\scriptsize$(n_2-j)(n_1-i+1)$};
        
        \node [left] at (0,0) {\scriptsize$0$};
        \node [left] at (0,1) {$\vdots$};
        \node [left] at (0,2) {\scriptsize$j$};
        \node [left] at (0,3) {$\vdots$};
        \node [left] at (0,4) {\scriptsize$n_2-1$};
        \node [below] at (0,0) {\scriptsize$0$};
        \node [below] at (2,0) {$\dots$};
        \node [below] at (3,0) {\scriptsize$i-1$};
        \node [below] at (4,0) {\scriptsize$i$};
        \node [below] at (5,0) {$\dots$};
        \node [below] at (6,0) {\scriptsize$n_1-1$};
        \end{tikzpicture}
        \caption{Sets $\Delta_{i,j}^{3,\sigma}$}
    \end{subfigure}
    \caption{Sets $\Delta_{i,j}^3$ and $\Delta_{i,j}^{3,\sigma}$ in Proposition \ref{rectwithp}}
    \label{fig:rectwithp}
\end{figure}

\begin{prop}\label{rectwithp}
With notation as in Proposition \ref{rect}, consider the family of subsets of $E$
$$\Delta=\Delta_{i,j}^3:=\left\{(e_1,e_2) \mid 0\leq e_1 \leq i \textrm{, } 0\leq e_2 \leq j-1\right\}\cup\{(0,j)\},$$
where $1\leq i\leq n_1-2$ and  $\max\left\{1,\frac{i(n_2+1)-n_1}{i}\right\} \leq j \leq n_2-2$ (see Figure \ref{fig:rectwithp} (1)).

Then, the MCCs, $C_\Delta^P$, are optimal $(r,\delta)=(i+1,n_1-i)$-LRCs.

Analogously, the MCCs, $C_\Delta^P$, where
$$\Delta=\Delta_{i,j}^{3,\sigma}:=\{(e_1,e_2) \mid 0\leq e_1 \leq i-1 \textrm{, } 0\leq e_2 \leq j\}\cup\{(i,0)\}\subseteq E,$$
$1\leq j\leq n_2-2$, and $\max\left\{1,\frac{j(n_1+1)-n_2}{j}\right\} \leq i \leq n_1-2$ (see Figure \ref{fig:rectwithp} (2)) are optimal $(r,\delta)=(j+1,n_2-j)$-LRCs.
\end{prop}

\begin{proof}
As before, we only give the proof for the case $\Delta=\Delta_{i,j}^3$ since a proof for $\Delta_{i,j}^{3,\sigma}$ follows as described in the symmetric situation of the proof of Propositon \ref{rectelim}.

$\Delta$ is obtained by removing the points $(e_1,j)$, $1\leq e_1 \leq i$, of a rectangle
$$\Delta_{i,j}=\left\{(e_1,e_2) \mid 0\leq e_1 \leq i \textrm{, } 0\leq e_2\leq j\right\}$$
with $1\leq i\leq n_1-2$ and $1\leq j\leq n_2-2$ such that $\dis(0,j)\leq\dis(i,j-1)$. As a consequence, $n_1(n_2-j)\leq (n_1-i)(n_2-j+1)$, which is equivalent to $i\leq \frac{n_1}{n_2-j+1}$, or $j\geq \frac{i(n_2+1)-n_1}{i}$. In this case, we interpolate with respect to $X$ and the parameters of the code $C_\Delta^P$ are $k=(i+1)j+1$, $d_0=n_1(n_2-j)$, $r=i+1$ and $\delta-1=n_1-i-1$. Thus,
\begin{equation*}
\begin{aligned}
k+d_0+\left(\ceil*{\frac{k}{r}}-1\right)(\delta-1)&=(i+1)j+1+n_1(n_2-j)+\left(\ceil*{\frac{(i+1)j+1}{i+1}}-1\right)(n_1-i-1)\\
&=n_1n_2+1.
\end{aligned}
\end{equation*}
\end{proof}

\begin{remark}\label{nomore}
    The families of (decreasing) MCCs given in Propositions \ref{rect}, \ref{rectelim} and \ref{rectwithp} determine the parameters of all $d_0$-optimal bivariate ($m=2$) $(r,\delta)$-LRCs $C_\Delta^P$ (with any set $\Delta\subseteq E$). That is to say, if $C_\Delta^P$ is a $d_0$-optimal LRC, then there exists an MCC, $C_{\Delta^*}^P$, as in Propositions \ref{rect}, \ref{rectelim} and \ref{rectwithp} having the same parameters $n$, $k$, $d$, $r$ and $\delta$ as $C_\Delta^P$. We omit the proof to shorten this article since our aim is to find optimal LRCs. Therefore, by Remark \ref{charact}, we have characterized the optimal bivariate decreasing MCCs.
\end{remark}

As a consequence of Remark \ref{nomore}, the next Corollary \ref{colpar} determines the parameters and $(r,\delta)$\hyp{}localities of the optimal $(r,\delta)$-LRCs we can obtain with the bound $d_0$ on the minimum distance. Notice that, in order not to repeat cases and since the variables $X$ and $Y$ play the same role, the parameters are written only with the notation we have used to interpolate with respect to $X$.

\begin{col}\label{colpar}
Let $\; \mathbb{F}_q$ be a finite field. For each pair $(n_1,n_2)$ of integers such that $\; 2\leq n_1,n_2\leq q$, there exists an optimal $(r,\delta)$-LRC with length $n=n_1n_2$, parameters $[n,k,d]_q$ and locality $(r,\delta)$ as follows:

\begin{enumerate}
    \item $k=(i+1)(j+1)$, $d=(n_1-i)(n_2-j)$, where
    \begin{itemize}
        \item $i=0$ and $0\leq j\leq n_2-1$, being the locality $(r,\delta)=(1,n_1)$; or
        
        \item $1\leq i\leq n_1-2$ and $j=n_2-1$, being the locality $(r,\delta)=(i+1,n_1-i)$.
    \end{itemize}
    
    \item $k=(i+1)(n_2-1)+s+1$, $d=n_1-s$ and $(r,\delta)=(i+1,n_1-i)$, where
    $$\max\left\{0,2i-n_1\right\}\leq s<i\leq n_1-2.$$
    
    \item $k=(i+1)j+1$, $d=n_1(n_2-j)$ and $(r,\delta)=(i+1,n_1-i)$, where $1\leq i\leq n_1-2$ and $\max\left\{1,\frac{i(n_2+1)-n_1}{i}\right\} \leq j \leq n_2-2$.
\end{enumerate}
\end{col}

\subsection{The case $m\geq 3$}\label{subm3}

In Subsection \ref{subm2} we have studied bivariate codes $C_\Delta^P$, obtained from decreasing sets $\Delta\subseteq \{0,1,\dots,n_1-1\}\times \{0,1,\dots,n_2-1\}$, which give rise to optimal LRCs. Moreover we have determined all the parameters of the $d_0$-optimal bivariate MCCs. We devote this subsection to the same purpose in the multivariate case. Thus $\mathcal{R}=\faktor{\mathbb{F}_q[X_1,\dots,X_m]}{I}$, where $m\geq 3$ and $\Delta\subseteq \{0,1,\dots,n_1-1\}\times \cdots \times \{0,1,\dots,n_m-1\}$. The forthcoming Propositions \ref{hyperrect} and \ref{hyperrectelim} are the analogs to Propositions \ref{rect} and \ref{rectelim} for multivariate MCCs and allow us to determine the parameters of the $d_0$-optimal LRCs of the type $C_\Delta^P$, $m\geq 3$.

\begin{prop}\label{hyperrect}
Keep the notation as given at the beginning of Section \ref{framework}. For each index $j_0\in\{1,\dots,m\}$, set $i_j=n_j-1$ for all $j\in\{1,\dots,m\}\backslash\{j_0\}$ and $i_{j_0}\in\left\{0,1,\dots,n_{j_0}-2\right\}$, and consider
$$\Delta=\Delta^1_{i_1,\dots,i_m}:=\left\{(e_1,\dots,e_m) \mid 0\leq e_j \leq i_j \textrm{, for all } j=1,\dots,m\right\}.$$
Then, the MCC, $C_\Delta^P$, is an optimal LRC with locality $(r,\delta)=\left(i_{j_0}+1,n_{j_0}-i_{j_0}\right)$. Furthermore, $\Delta^1_{i_1,\dots,i_m}$ are the unique sets of the form $\Delta'=\{(e_1,\dots,e_m) \mid 0\leq e_j\leq l_j \textrm{ for all } j=1,\dots,m\}$, where $0\leq l_j\leq n_j-1$, providing optimal LRCs.
\end{prop}

\begin{proof}
We interpolate with respect to $X_1$ (the proof is analogous if we interpolate with respect to any other variable). Consider a set $\Delta'$ as in the statement.

We start by assuming that $l_j=n_j-1$ for $m-2$ indices $j$. Without loss of generality suppose that $l_j=n_j-1$ for all $j=3,\dots,m$. Then, the point that defines the bound on the minimum distance is $(l_1,l_2,n_3-1,\dots,n_m-1)$ and the parameters of this code give the following value for the LHS of (\ref{formula2}):
\begin{equation*}
    \begin{aligned}
        d_0+k+\left(\ceil*{\frac{k}{r}}-1\right)(\delta-1)&=(n_1-l_1)(n_2-l_2)+(l_1+1)(l_2+1)n_3n_4\dots n_m\\
        &+[(l_2+1)n_3n_4\dots n_m-1](n_1-l_1-1)\\
        &=(n_1-l_1)(n_2-l_2)+n_1(l_2+1)n_3n_4\dots n_m-(n_1-l_1-1)\\
        &=n_1(l_2+1)n_3n_4\dots n_m+(n_1-l_1)(n_2-l_2-1)+1.
    \end{aligned}
\end{equation*}
Thus, the code is optimal if and only if $l_2=n_2-1$ (and $l_1\in\{0,1,\dots,n_1-2\}$ for being an LRC).

We conclude the proof after noticing that the same reasoning allows us to prove the proposition when the number of indices $j$ in $\Delta'$ such that $l_j=n_j-1$ is less than $m-2$.
\end{proof}

Our next result shows that deleting, from a set $\Delta^1_{i_1,\dots,i_m}$, a suitable number of successive minimum distance points on the line $e_j=n_j-1$, $j\neq j_0$, an optimal LRC is also obtained. This is because for each removed point we lose one unit in dimension but we gain one unit in the bound for the minimum distance and $r$, $\delta$ and $\ceil*{\frac{k}{r}}$ do not change. As a consequence the LHS in (\ref{formula2}) remains constant.

\begin{prop}\label{hyperrectelim}
Keep the notation as in Proposition \ref{hyperrect}. Define
$$\Delta=\Delta_{i_{j_0},s}^2:=\Delta^1_{i_1,\dots,i_m}\big\backslash \left\{\left(n_1-1,\dots,n_{j_0-1}-1,e_{j_0},n_{j_0+1}-1,\dots,n_m-1\right) \mid s\leq e_{j_0}\leq i_{j_0}\right\},$$
where $s$ satisfies $\max\left\{1,2i_{j_0}-n_{j_0}+1\right\}\leq s \leq i_{j_0}\leq n_{j_0}-2$ or $i_{j_0}=s=0$.

Then the MCC, $C_\Delta^P$, is an optimal LRC with locality $(r,\delta)=\left(i_{j_0}+1,n_{j_0}-i_{j_0}\right)$.
\end{prop}

\begin{proof}
The distance $\dis(\mathbf{p})$ (see Definition \ref{dist}) of the point
$$\mathbf{p}=(n_1-1,n_2-1,\dots,n_{j_0-1}-1,i_{j_0},n_{j_0+1}-1,\dots,n_{m-1}-1,n_m-1)$$
determines the bound $d_0$ for the minimum distance of the code $C_{\Delta_{i_1,\dots,i_m}^1}^P$. We look for an index $0 \leq s\leq i_{j_0}$ such that $i_{j_0}-s+1$ is the number of points in $\Delta_{i_1,\dots,i_m}^1$ that meet the line $e_j=n_j-1$, $j\neq j_0$, and have distance less than $2\left(n_{j_0}-i_{j_0}\right)$. The candidate set $\Delta$ for $C_\Delta^P$ to be optimal is obtained by deleting from $\Delta_{i_1,\dots,i_m}^1$ those points because $2\left(n_{j_0}-i_{j_0}\right)$ is the distance of any point in the set
$$V=\left\{\mathbf{p}-\boldsymbol{\epsilon}_j \textrm{ for all } j\in\{1,\dots,m\}\backslash\{j_0\}\right\},$$
where $\boldsymbol{\epsilon}_j=(\delta_{j1},\dots,\delta_{jm})$, $\delta_{ij}$ being the Kronecker delta, and $V\subseteq \Delta^1_{i_1,\dots,i_m} \big\backslash \{\mathbf{p}\}$. Thus, $n_{j_0}-s< 2(n_{j_0}-i_{j_0})$, what is equivalent to $s\geq 2i_{j_0}-n_{j_0}+1$.

Therefore, in order to $\Delta$ be a candidate for $C_\Delta^P$ to be optimal, $s\geq\max \{0,2i_{j_0}-n_{j_0}+1\}$. The dimension of the code $C_\Delta^P$ is
$$k=n_1n_2\cdots n_{j_0-1}(i_{j_0}+1)n_{j_0+1}\cdots n_{m-1}n_m-(i_{j_0}-s+1),$$
and the bound on the minimum distance of $C_\Delta^P$ is given by the point with coordinates $e_j=n_j-1$, $j\neq j_0$, $e_{j_0}=s-1$ when $s\geq 1$ or by any point of $V$ when $s=0$. Then $d_0=n_{j_0}-s+1$ for $s\geq 1$ and $d_0=2(n_{j_0}-i_{j_0})$ when $s=0$. Moreover we interpolate with respect to $X_{j_0}$ (it is the only way to obtain an LRC), so $r=i_{j_0}+1$ and $\delta-1=n_{j_0}-i_{j_0}-1$. Thus, the value for $k+d_0+\left(\ceil*{\frac{k}{r}}-1\right)(\delta-1)$ (the LHS of (\ref{formula2})) is
\begin{equation*}
    \begin{array}{ll}
    & n_1n_2\cdots n_{j_0-1}(i_{j_0}+1)n_{j_0+1}\cdots n_{m-1}n_m-(i_{j_0}-s+1)+n_{j_0}-s+1\\
    & +\left(\ceil*{\frac{n_1n_2\cdots n_{j_0-1}(i_{j_0}+1)n_{j_0+1}\cdots n_{m-1}n_m-(i_{j_0}-s+1)}{i_{j_0}+1}}-1\right) \cdot (n_{j_0}-i_{j_0}-1)\\
    =&n_1n_2\cdots n_m -i_{j_0}+n_{j_0}-(n_{j_0}-i_{j_0}-1)=n_1n_2\cdots n_m +1, \text{ if } s\geq 1 \text{ and}\\
    \\
    & n_1n_2\cdots n_{j_0-1}(i_{j_0}+1)n_{j_0+1}\cdots n_{m-1}n_m-(i_{j_0}-s+1)+2(n_{j_0}-i_{j_0})\\
    & +\left(\ceil*{\frac{n_1n_2\cdots n_{j_0-1}(i_{j_0}+1)n_{j_0+1}\cdots n_{m-1}n_m-(i_{j_0}-s+1)}{i_{j_0}+1}}-1\right)\cdot (n_{j_0}-i_{j_0}-1)\\
    =& n_1n_2\cdots n_m -i_{j_0}-1+2(n_{j_0}-i_{j_0})-2(n_{j_0}-i_{j_0}-1)=n_1n_2\cdots n_m +1-i_{j_0}, \text{ otherwise,}
    \end{array}
\end{equation*}
which proves that $C_\Delta^P$ is optimal and concludes the proof.
\end{proof}

\begin{remark}\label{gnomore}
    As in the bivariate case, the families of (decreasing) MCCs given in Propositions \ref{hyperrect} and \ref{hyperrectelim} determine the parameters of all $d_0$-optimal multivariate ($m\geq 3$) $(r,\delta)$-LRCs $C_\Delta^P$ (with any set $\Delta\subseteq E$). Again we omit the proof, which follows from a close reasoning to that of the bivariate case. Therefore, by Remark \ref{charact}, we have characterized the optimal multivariate decreasing MCCs.
\end{remark}

Corollary \ref{gcolpar} determines parameters and $(r,\delta)$-localities of the multivariate $d_0$-optimal $(r,\delta)$-LRCs.

\begin{col}\label{gcolpar}
Let $\mathbb{F}_q$ be a finite field and consider an integer $m\geq 3$. For every $m$-tuple $(n_1,\dots,n_m)$ of integers such that $2 \leq n_j \leq q$, $j\in\{1,\dots,m\}$, there exists an optimal $(r,\delta)$-LRC with length $n=n_1\cdots n_m$, parameters $[n,k,d]_q$ and locality $(r,\delta)$ as follows:

\begin{enumerate}
    \item $k=n_1\cdots n_{j_0-1}(i_{j_0}+1)n_{j_0+1}\cdots n_m$, $d=n_{j_0}-i_{j_0}$ and $(r,\delta)=(i_{j_0}+1,n_{j_0}-i_{j_0})$, where $i_{j_0}\in\{0,1,\dots,n_{j_0}-2\}$.
    
    \item $k=n_1\cdots n_{j_0-1}(i_{j_0}+1)n_{j_0+1}\cdots n_m-(i_{j_0}-s+1)$, $d=n_{j_0}-s+1$ and $(r,\delta)=(i_{j_0}+1,n_{j_0}-i_{j_0})$, where
    $$\max\left\{1,2i_{j_0}-n_{j_0}+1\right\}\leq s\leq i_{j_0}\leq n_{j_0}-2.$$
    
    \item $k=n_1\cdots n_{j_0-1}n_{j_0+1}\cdots n_m-1$, $d=2n_{j_0}$ and $(r,\delta)=(1,n_{j_0})$.
\end{enumerate}
\end{col}

\begin{remark}
Keep the notation as in Section \ref{framework}, so let $m\geq 2$. Let $\mathbb{N}$ be the set of nonnegative integers and $\Delta$ be a subset of $E$ satisfying some of the conditions in Propositions \ref{rect}, \ref{rectelim}, \ref{rectwithp}, \ref{hyperrect} or \ref{hyperrectelim}. Define $\Delta^*:=\boldsymbol{v}+\Delta$ for any $\boldsymbol{v}\in\mathbb{N}^m$ such that $\Delta^*\subseteq E$. If $0\notin P_j$ for all $1\leq j \leq m$, then the MCC $C^P_{\Delta^*}$ is optimal with the same parameters and locality as $C^P_\Delta$. This result follows straightforwardly from Remark \ref{rempseudo}.
\end{remark}

\begin{remark}\label{comp}
MCCs include the family of codes introduced in \cite{ACNV2021}, codes whose evaluation map is the same as MCCs but their evaluation sets $V_\Delta$ are only a subset of those used for MCCs. Specifically, the codes in \cite{ACNV2021} are subcodes of affine cartesian codes (of order $d$), where the corresponding set $V_\Delta$ is the set of polynomials $f$ in $\mathbb{F}_q[X_1,\dots,X_m]$ with total degree bounded by $d$ and such that a fixed variable $X_{j_0}$ has degree $\deg_{X_{j_0}}(f)\leq i_{j_0}<n_{j_0}-1$ for some fixed integer $i_{j_0}$ (see \cite[Definitions 2.2 and 2.3]{ACNV2021}). Therefore, while MCCs allow arbitrary sets $\Delta\subset E$, the sets $\Delta$ of those codes considered in \cite{ACNV2021} are of the form
$$\Delta=\Delta_{j_0}=\{(e_1,\dots,e_m)\in E \mid e_1+\cdots+e_m\leq d, e_{j_0}\leq i_{j_0}\}.$$
As a consequence we obtain many more $(r,\delta)$-optimal LRCs than those given in \cite[Corollaries 4.2 and 4.3]{ACNV2021}. Thus, if we fix the locality $r=i_{j_0}+1$ for some $1\leq j_0 \leq m$, then we obtain optimal codes which are not considered in \cite{ACNV2021}. These are those of Proposition \ref{rect} for $i_{j_0}=i=0$, $j<n_2-2$, and $i<n_1-2$, $i_{j_0}=j=0$; those of Proposition \ref{rectelim} for $s\leq i_{j_0}-2$; those of Proposition \ref{rectwithp} 
for $i_{j_0}>1$ and for $i_{j_0}=1$ and $n_{j_0}<n_{j'}$, where $j'\in\{1,2\}\backslash\{j_0\}$; and those of Proposition \ref{hyperrectelim} for $i_{j_0}\geq 2$ and $\max\{1,2i_{j_0}-n_{j_0}\}\leq s \leq i_{j_0}-1$. Moreover, in this paper, we also give many more optimal LRCs, regarded as subfield-subcodes of MCCs, as we will explain in the next section.
\end{remark}

\section{Optimal subfield-subcodes}\label{subfield}

$J$-affine variety codes were introduced in \cite{GHR2015} and they are a subclass of MCCs. We devote this section to prove that subfield-subcodes of some $J$-affine variety codes keep the parameters and $(r,\delta)$-locality of certain decreasing MCCs, giving rise to new $(r,\delta)$-LRCs over smaller supporting fields. In fact, in this section we provide optimal LRCs with parameters that cannot be found in the literature \cite{CXHF2018,LMX2019,LXY2019,SZW2019,J2019,CFXF2019,XC2019,CH2019,FF2020,Z2020,ZL2020,QZF2021,CWLX2021,LEL2021,KWG2021}. Our LRCs are $p^h$-ary, $p$ a prime, such that $r+\delta-1$ equals either $p^h+1$ or $p^h+2$, their length $n$ is a multiple of some of these two values, $r>1$ and $\delta>2$. On the contrary, the codes given in the literature satisfy:
\begin{itemize}
    \item $r+\delta-1\leq p^h$ \cite{CH2019,FF2020,Z2020,ZL2020,CWLX2021,KWG2021};
    \item $r+\delta-1\leq p^h+1$ with either minimum distances other than ours \cite{SZW2019,LEL2021} or opposite gcd-type conditions \cite{SZW2019}, see Remarks \ref{remsub} and \ref{remsub2};
    \item either $n\mid p^h-1$ or $n\mid p^h+1$ \cite{CXHF2018,CFXF2019,QZF2021}, but our codes have $n\geq 2(p^h+1)$;
    \item $r=1$ \cite{XC2019};
    \item $\delta=2$ \cite{LMX2019,LXY2019,J2019,KWG2021}; and
    \item $2\delta+1\leq d\leq r+\delta$ \cite{KWG2021} but, in case our codes have $d\leq r+\delta$, then $d\leq 2\delta$, see Remarks \ref{remsub} and \ref{remsub2}.
\end{itemize}

Subfield-subcodes of $J$-affine variety codes were also used in \cite{GHM2020} to provide $(r,\delta)$-LRCs, most of them non-optimal. However the recovery procedure in \cite{GHM2020} was different and the obtained codes were distinct of those in this section.

\subsection{Subfield-subcodes}\label{factssubfield}

In this subsection we recall some facts about subfield-subcodes which will be useful in the forthcoming subsections. We keep the notation as in Section \ref{framework}. Assume that $q=p^l$, where $p$ is a prime number and $l\geq 2$. Pick a positive integer $h$ such that $h\mid l$ and regard $\mathbb{F}_{p^h}$ as a subfield of $\mathbb{F}_q=\mathbb{F}_{p^l}$. Consider a subset $J\subseteq \{1,\dots,m\}$ and assume that the polynomials $f_j(X_j)$ generating the ideal $I$ are of the form $$f_j(X_j)=X_j^{n_j}-1,$$
for some $n_j\mid q-1$ if $j\in J$, and
$$f_j(X_j)=X_j^{n_j}-X_j,$$
where $n_j-1 \mid q-1$, otherwise. Then, each set $P_j\subseteq \mathbb{F}_q$ introduced in Section \ref{framework} is the set of $n_j$-th roots of unity if $j\in J$ or the set of $n_j-1$-th roots of unity together with 0 otherwise. The corresponding MCC is denoted by $C_\Delta^{P,J}$. As introduced in \cite{GHR2015}, $C_\Delta^{P,J}$ is a $J$-\textit{affine variety code}.

\begin{mydef}
The linear code $S_\Delta^{P,J}:=C_\Delta^{P,J}\cap \mathbb{F}_{p^h}^n$ is the \textit{subfield-subcode over the field $\mathbb{F}_{p^h}$} of $C_\Delta^{P,J}$.
\end{mydef}

When $j\notin J$, the evaluation of monomials containing $X_j^0$ or containing $X_j^{n_j-1}$ may be different (see \cite{GGHR2017} for details). This explains the difference on the powers on the variables when equipping $E=\{0,1,\dots,n_1-1\}\times\dots\times\{0,1,\dots,n_m-1\}$ with the following structure which we will assume in the sequel. When $j\in J$ then we identify the set $\{0,1,\dots,n_j-1\}$ with the ring $\mathbb{Z}/n_j\mathbb{Z}$. Otherwise, if $j\notin J$, we identify the set $\{1,\dots,n_j-1\}$ with $\mathbb{Z}/(n_j-1)\mathbb{Z}$, and we extend the addition and multiplication in this ring to $\{0,1,\dots,n_j-1\}$, by setting $0+e=e$, $0\cdot e=0$ for all $e=0,1,\dots,n_j-1$. Therefore, $\{0,1,\dots,n_j-1\}=\{0\}\cup\mathbb{Z}/(n_j-1)\mathbb{Z}$.

We call a set $\Omega\subseteq E$ a \textit{cyclotomic set} with respect to $p^h$ if $p^h\boldsymbol{\omega} \in \Omega$ for all $\boldsymbol{\omega}=(\omega_1,\dots,\omega_m)\in \Omega$. Minimal cyclotomic sets are those of the form $\Lambda=\{p^{hi}\mathbf{e} \mid i\geq 0\}$, for some element $\mathbf{e}\in E$. In this paper we will refer to cyclotomic sets as \textit{closed} sets since they are unions of minimal cyclotomic sets. For each minimal closed set $\Lambda$, denote by $\mathbf{x}$ the minimum element in $\Lambda$ with respect to the lexicographic order and set $\Lambda=\Lambda_\mathbf{x}$. Hence, $\Lambda_\mathbf{x}=\{\mathbf{x},p^h\mathbf{x},\dots,p^{h(\#\Lambda_\mathbf{x}-1)}\mathbf{x}\}$. Fixed an index $j\in\{1,\dots,m\}$, if we replace $E$ by $\{0,1,\dots,n_j-1\}$, the same definition gives rise to sets $\Omega^j\subseteq \{0,1,\dots,n_j-1\}$ (respectively, $\Lambda^j\subseteq \{0,1,\dots,n_j-1\}$) called \textit{closed} (respectively, \textit{minimal closed) sets in a single variable} with respect to $p^h$. Again, denoting by $x$ the minimum element in $\Lambda^j$, we set $\Lambda^j=\Lambda^j_x$. For example, assume $m=3$, $p=2$, $h=1$, $l=3$, $J=\emptyset$, $n_1=n_2=n_3=8$ and $\mathbf{x}=(1,4,5)$. Then, it holds that $\Lambda_\mathbf{x}=\{(1,4,5),(2,1,3),(4,2,6)\}\subseteq E=\{0,1, \dots,7\}^3$ ($E$ has the same structure as $\left(\{0\}\cup\mathbb{Z}/7\mathbb{Z}\right)^3$) is a minimal closed set with respect to $2$ and the corresponding minimal closed set in a single variable for $j=3$ is the set $\Lambda^j_{3}=\{3,5,6\}\subseteq \{0,1,\dots,7\}$ ($\{0,1,\dots,7\}$ is identified with the ring $\{0\}\cup\mathbb{Z}/7\mathbb{Z}$).

Now, we define three trace type maps which will be useful: $\tr_l^h \colon \mathbb{F}_{p^l} \to \mathbb{F}_{p^h}$, $\tr_l^h(x)=x+x^{p^h}+\cdots+x^{p^{h\left(\frac{l}{h}-1\right)}}$; $\btr \colon \mathbb{F}_{p^l}^n \to \mathbb{F}_{p^h}^n$, determined by $\tr_l^h$ componentwise and $\mathcal{T} \colon \mathcal{R} \to \mathcal{R}$, $\mathcal{T}(f)=f+f^{p^h}+\cdots+f^{p^{h \left(\frac{l}{h}-1\right)}}$, $\mathcal{R}$ being the quotient ring defined at the beginning of Section \ref{framework}. Recall from Section \ref{defs}, that the projection map $\mathbb{F}_q^n \to \mathbb{F}_q^r$ on the coordinates of a subset $R\subseteq \{1,\dots,n\}$ of cardinality $r$ is denoted by $\pi_R$.

Next result shows that when $\Delta$ is closed, then the operators on a code ``taking its projection'' and ``taking its subfield-subcode'' commute.

\begin{prop}\label{commute}
    With notation as in Section \ref{defs}, let $R\subseteq\{1,\dots,n\}$. If $\Delta$ is closed, then $\pi_R(S_\Delta^{P,J})=\pi_R(C_\Delta^{P,J}) \cap \mathbb{F}_{p^h}^{\#R}$.
\end{prop}

\begin{proof}
First we prove that $S_\Delta^{P,J}=\normalfont{\btr}(C_\Delta^{P,J})$. By reasoning as in Propositions 4 and 5 of \cite{GH2015}, it holds the following chain of equalities:
\begin{equation*}
    \begin{aligned}
        & \btr(C_\Delta^{P,J})=\left\{\btr(\mathbf{c}) \mid \mathbf{c}\in C_\Delta^{P,J}\right\}=\left\{\btr(\ev_P(f)) \mid f\in \mathcal{R} \text{, } \supp(f)\subseteq \Delta\right\}=\\
        & \left\{\ev_P(\mathcal{T}(f)) \mid f\in \mathcal{R} \text{, } \supp(f)\subseteq \Delta\right\}=\left\{\ev_P(\mathcal{T}(f)) \mid f\in \mathcal{R} \text{, } \supp(\mathcal{T}(f))\subseteq \Delta\right\}=S_\Delta^{P,J}.
    \end{aligned}
\end{equation*}
Notice that the last but one equality is true because $\Delta$ is closed. Now, define $\btr' \colon \mathbb{F}_{p^l}^{\#R} \to \mathbb{F}_{p^h}^{\#R}$, determined by $\tr_l^h$ componentwise. Then, $\pi_R(C_\Delta^{P,J}) \cap \mathbb{F}_{p^h}^{\#R}=\btr'(\pi_R(C_\Delta^{P,J}))$. Finally, for any element in $S_\Delta^{P,J}$, $\btr(\mathbf{c})$, $\mathbf{c}\in C_\Delta^{P,J}$, the fact that the maps $\btr$ and $\btr'$ are defined componentwise implies $\pi_R(\btr(\mathbf{c}))=\btr'(\pi_R(\mathbf{c}))$, which proves the result.
\end{proof}

Closed sets will be the key for obtaining optimal $(r,\delta)$-LRCs coming from subfield\hyp{}subcodes. To explain it, we recall, on the one hand, that if $\Delta$ is a closed set, then $\dim(S_\Delta^{P,J})=\dim(C_\Delta^{P,J})=\#\Delta$ \cite[Theorem 2.3]{GHM2020}. On the other hand, the minimum distance of a subfield-subcode $S_\Delta^{P,J}$ admits the bound on the MCC $C_\Delta^{P,J}$ it comes from. Since $\Delta$ is closed, it is not decreasing. Therefore, the bound given in Corollary \ref{pdist} is not sharp, which forces us to use an improved bound for each particular case. This bound coincides with the one on a certain decreasing MCC $C_{\Delta'}^{P,J}$ obtained, roughly speaking, after compacting $\Delta$ so that we remove gaps in $\Delta$ to obtain a decreasing set $\Delta'$ such that $\#\Delta=\#\Delta'$. Thus, if we choose $\Delta$ to be closed, the code over $\mathbb{F}_{p^h}$, $S_\Delta^{P,J}$, has the same parameters $n$ and $k$ and the same bound for the minimum distance as $C_{\Delta'}^{P,J}$. Moreover, the recovery method presented in Proposition \ref{method} can also be applied to $S_\Delta^{P,J}$ obtaining the same locality $(r,\delta)$ as $C_{\Delta'}^{P,J}$.

\subsection{Optimal $(r,\delta)$-LRCs coming from subfield-subcodes of bivariate MCCs}\label{opbivsubfield}

In this subsection, we use some results in Section \ref{work} and the ideas described in the above paragraph to provide some families of new optimal $(r,\delta)$-LRCs coming from subfield-subcodes of bivariate $J$-affine variety codes. We will give $p^h$-ary optimal $(r,\delta)$-LRCs whose length is a multiple of $r+\delta-1$, where $r+\delta-1$ equals $p^h+1$ or $p^h+2$, $r>1$, $\delta>2$ and for some codes we impose certain gcd-type conditions so that all the codes provided are new (see the introduction of Section \ref{subfield} and the future Remark \ref{remsub}). The forthcoming Propositions \ref{bivq+1} and \ref{bivq+2} (in characteristic two) prove the optimality while Theorems \ref{thmparbivq+1} and \ref{thmparbivq+2} show the parameters of our codes.

Let $U_t\subseteq \mathbb{F}_q$ denote the set of $t$-th roots of unity, $t\mid q-1$. Keep the notation as in Section \ref{framework} and Subsection \ref{factssubfield}. Fix $i\in\{1,2\}$ (it refers to the variable $X_i$ with respect to which we will interpolate) and denote $i'$ the unique element $i'\in\{1,2\}\backslash\{i\}$.

Pick $p^h\geq 4$ if $p$ equals $2$ ($p^h\geq 5$, otherwise) such that $p^h+1\mid q-1$ and set $P_i=U_{p^h+1}\subseteq\mathbb{F}_q$, then $n_i=p^h+1$. Our set $P$ is $P=P_1\times P_2$, where $P_{i'}$ is either $U_{n_{i'}}\subseteq\mathbb{F}_q$, with $n_{i'}\mid q-1$ and $J=\{1,2\}$, or $U_{n_{i'}-1}\cup\{0\}\subseteq\mathbb{F}_q$, with $n_{i'}-1\mid q-1$ and $J=\{i\}$.

The following two families of sets will be used to define the sets $\Delta$ of our codes $S_\Delta^{P,J}$ since they will constitute the sets $\supp_{X_i}(V_\Delta)$ defined under Definition \ref{decr}. For each nonnegative integer $a\leq \floor*{\frac{p^h}{2}}-1$ (and, if $p=2$, $b\leq \frac{p^h}{2}-2$) define
$$\Omega_a:=\left\{0,1,\dots,a,p^h+1-a,p^h+2-a,\dots,p^h\right\} =\Lambda^i_0\cup \Lambda^i_1\cup\cdots\cup \Lambda^i_a$$
when $a>0$, $\Omega_0:=\{0\}$ and
$$\Omega^*_b:=\left\{\frac{p^h}{2}-b,\frac{p^h}{2}-b+1,\dots,\frac{p^h}{2}+b+1\right\}=\Lambda^i_{\frac{p^h}{2}-b}\cup\Lambda^i_{\frac{p^h}{2}-b+1}\cup \cdots \cup \Lambda^i_{\frac{p^h}{2}},$$
which are closed sets of $\{0,1,\dots,n_i-1\}=\{0,1,\dots,p^h\}$ (identified with $\mathbb{Z}/(p^h+1)\mathbb{Z}$) in the variable $i$ with respect to $p^h$. Indeed, with the identification, $p^h+1=0$ and then $\Lambda^i_0=\{0\}$ and $\Lambda^i_t=\{t,p^h-(t-1)\}$.

\begin{ex}\label{ex1}
Set $(i,p^h,q,a,b)=(1,8,64,3,2)$, then the above defined sets are $\Omega_a=\{0,1,2,3,6,7,8\}=\{0\}\cup\{1,8\}\cup\{2,7\}\cup\{3,6\}$ and $\Omega^*_b=\{2,\dots,7\}=\{2,7\}\cup\{3,6\}\cup\{4,5\}$, and they coincide, respectively, with the set $\supp_{X_i}(V_\Delta)$ in Figure \ref{fig:proof1} a) (\textsc{i}) and b) (\textsc{i}).
\end{ex}

Now, let $0\leq t<z\leq \floor*{\frac{p^h}{2}}-1$ be nonnegative integers such that $2t\geq \max\{0,4z-p^h-1\}$. In addition, when $p=2$, consider a nonnegative integer $0\leq u\leq \frac{p^h}{2}-2$ and if $u\geq 1$, let $0\leq v<u$ be a nonnegative integer such that $2v+1\geq \max\{0,4u+1-p^h\}$.
Define
    $$\Delta_1(z)=\Delta_1:=\begin{cases}
        \Omega_z\times\{0,1,\dots,n_2-1\}, & \text{when } i=1,\\
        \{0,1,\dots,n_1-1\}\times \Omega_z, & \text{otherwise;}
    \end{cases}$$
    $$\Delta_2(z,t)=\Delta_2:=\begin{cases}
        \Omega_z\times\{0,1,\dots,n_2-2\} \cup \Omega_t\times\{n_2-1\}, & \text{when } i=1,\\
        \{0,1,\dots,n_1-2\}\times \Omega_z \cup \{n_1-1\} \times \Omega_t, & \text{otherwise;}
    \end{cases}
    $$
    $$\Delta^*_1(u)=\Delta^*_1:=\begin{cases}
        \Omega^*_u\times\{0,1,\dots,n_2-1\}, & \text{when } i=1,\\
        \{0,1,\dots,n_1-1\}\times \Omega^*_u, & \text{otherwise;}
    \end{cases}
    $$
and
    $$\Delta^*_2(u,v)=\Delta^*_2:=\begin{cases}
        \Omega^*_u\times\{0,1,\dots,n_2-2\} \cup \Omega^*_v\times\{n_2-1\}, & \text{when } i=1,\\
        \{0,1,\dots,n_1-2\}\times \Omega^*_u \cup \{n_1-1\} \times \Omega^*_v, & \text{otherwise.}
    \end{cases}
    $$

\begin{ex}
This is a continuation of Example \ref{ex1}. With the same notation, set $z=a$ and $t=b$, consider also $(u,v)=(2,1)$. Then, $\Omega_t=\{0,1,2,7,8\}$ and $\Omega^*_v=\{3,\dots,6\}$. Figure \ref{fig:proof1} c) (\textsc{i}) and d) (\textsc{i}) show, respectively, the sets $\Delta_2$ and $\Delta^*_2$ in this case.   
\end{ex}

\begin{lemma}\label{lemmaph+1}
Keep the above notation. Let $a\leq \floor*{\frac{p^h}{2}}-1$ and, if $p=2$, $b\leq \frac{p^h}{2}-2$ be nonnegative integers. Consider the $\mathbb{F}_q$-vector spaces $V_1=\langle (X_i)^e \mid e \in \Omega_a \rangle$ and $V_2=\langle (X_i)^e \mid e \in \Omega^*_b \rangle$ contained in the quotient ring $\mathcal{R}_i$ defined before Proposition \ref{method}. Then, $\ev_{P_i}(V_1)$ and $\ev_{P_i}(V_2)$ are MDS codes.
\end{lemma}

\begin{proof}
Let $\Omega:=\{0,1,\dots,2a\}=\Omega_a+a$ regarded as representatives of elements in $\mathbb{Z}/(p^h+1)\mathbb{Z}$. Define $V=\langle (X_i)^e \mid e \in \Omega\rangle$. Codewords in $\ev_{P_i}(V)$ are of the form
$$\ev_{P_i}((X_i)^af)=\ev_{P_i}((X_i)^a)*\ev_{P_i}(f),$$
where $f\in V_1$. Since $0\notin P_i$, $\ev_{P_i}(V_1)$ and $\ev_{P_i}(V)$ are isometric codes. The code $\left(\ev_{P_i}(V)\right)^\perp$ is a $[p^h+1,p^h-2a,\leq 2a+2]_q$ code and, since $\Omega$ contains $2a+1$ consecutive elements, $\dis\left(\left(\ev_{P_i}(V)\right)^\perp\right)\geq 2a+2$ because its corresponding parity-check matrix contains a Vandermonde matrix of rank $2a+1$. Thus, $\left(\ev_{P_i}(V)\right)^\perp$ is an MDS code and therefore $\ev_{P_i}(V)$ and $\ev_{P_i}(V_1)$ are MDS codes. The fact that $\Omega^*_b$ contains $2b+2$ consecutive elements proves that $\left(\ev_{P_i}(V_2)\right)^\perp$ is an MDS code and therefore so is $\ev_{P_i}(V_2)$.\
\end{proof}

\begin{prop}\label{bivq+1}
Keep the the notation as above where $\mathbb{F}_{p^h}$ is regarded as a subfield of $\;\mathbb{F}_{q=p^l}$ and $p^h+1 \mid q-1$. Fixed $i$ and $P_i=U_{p^h+1}$, the set of $p^h+1$-th roots of unity, the following statements determine sets $P_{i'}$, $J$ and $\Delta$ such that the subfield-subcodes $S_\Delta^{P,J}$ over the field $\mathbb{F}_{p^h}$ are optimal $(r,\delta)$-LRCs.
\begin{enumerate}
    \item $P_{i'}=U_{n_{i'}}$ for some $n_{i'}$ such that $n_{i'}\mid q-1$; $J=\{1,2\}$ and $\Delta=\Delta_1$, in which case
    $$(r,\delta)=(2z+1,p^h-2z+1).$$
    
    \item $P_{i'}=U_{n_{i'}-1}\cup\{0\}$ for some $n_{i'}$ such that  $n_{i'}-1\mid q-1$; $J=\{i\}$ and $\Delta=\Delta_1$, in which case
    $$(r,\delta)=(2z+1,p^h-2z+1).$$

    \item $P_{i'}=U_{n_{i'}-1}\cup\{0\}$ for some $n_{i'}$ such that $n_{i'}-1\mid q-1$ and, if $p$ is odd, either $\gcd(n_{i'},p^h)\neq 1$ or $\gcd(n_{i'},p^h+1)\neq 1$; $J=\{i\}$ and $\Delta=\Delta_2$, in which case
    $$(r,\delta)=(2z+1,p^h-2z+1).$$
    
    \item $P_{i'}=U_{n_{i'}}$ for some $n_{i'}$ such that $n_{i'}\mid q-1$; $J=\{1,2\}$ and $\Delta=\Delta^*_1$, in which case
    $$(r,\delta)=(2u+2,p^h-2u).$$

    \item $P_{i'}=U_{n_{i'}-1}\cup\{0\}$ for some $n_{i'}$ such that $n_{i'}-1\mid q-1$; $J=\{i\}$ and $\Delta=\Delta^*_1$, in which case
    $$(r,\delta)=(2u+2,p^h-2u).$$

    \item $P_{i'}=U_{n_{i'}-1}\cup\{0\}$ for some $n_{i'}$ such that $n_{i'}-1\mid q-1$; $J=\{i\}$ and $\Delta=\Delta^*_2$, in which case $(r,\delta)=(2u+2,p^h-2u)$.
\end{enumerate}
\end{prop}

\begin{proof}
We start by proving that the sets $\Delta$ in the statements (1)-(6) are closed with respect to $p^h$. As we said, in the single variable $i$, the subsets of $\{0,1,\dots,n_i-1\}=\{0,1,\dots,p^h\}$ (identified with $\mathbb{Z}/(p^h+1)\mathbb{Z}$),
$$\Omega_a=\Lambda_0^i\cup \Lambda_1^i\cup \dots \cup \Lambda_a^i$$
and
$$\Omega^*_b=\Lambda_{\frac{p^h}{2}-b}^i\cup \Lambda_{\frac{p^h}{2}-b+1}^i\cup \dots \cup \Lambda_{\frac{p^h}{2}}^i,$$
for $a\in \{z,t\}$ and $b\in\{u,v\}$ are clearly closed. In the single variable $i'$, $\{0,1,\dots,n_{i'}-1\}$ is closed. In addition, when $0\in P_{i'}$, the minimal closed set in a single variable $\Lambda^{i'}_{n_{i'}-1}\subseteq \{0,1,\dots,n_{i'}-1\}$ is the set $\Lambda^{i'}_{n_{i'}-1}=\{n_{i'}-1\}$. Indeed, with the identification $n_{i'}=1$ described in Subsection \ref{factssubfield}, it holds the following chain of equalities:
$$p^h(n_{i'}-1)=(p^h-1)(n_{i'}-1)+n_{i'}-1=(p^h-1)n_{i'}+n_{i'}-p^h=p^h-1+n_{i'}-p^h=n_{i'}-1.$$
Therefore, $\{0,1,\dots,n_{i'}-2\}=\{0,1,\dots,n_{i'}-1\}\backslash \{n_{i'}-1\}$ is also closed. The cartesian product and the union of closed sets are closed, so the sets $\Delta$ in (1)-(6) are closed and $\dim(S_\Delta^{P,J})=\dim(C_\Delta^{P,J})$.

\medskip
Now we are going to prove that the subfield-subcodes $S_\Delta^{P,J}$ are LRCs. Let $i=1$ and $V_1$ as in Lemma \ref{lemmaph+1} with $a=z$. Since $\Omega_z$ is closed, $\dim(\ev_{P_i}(V_1)\cap \mathbb{F}_{p^h}^{p^h+1})=\dim(\ev_{P_i}(V_1))$ and the fact that $\dis(\ev_{P_i}(V_1)\cap \mathbb{F}_{p^h}^{p^h+1})\geq \dis(\ev_{P_i}(V_1))$ and Lemma \ref{lemmaph+1} imply that $\ev_{P_i}(V_1)\cap \mathbb{F}_{p^h}^{p^h+1}$ is an MDS code with minimum distance $p^h-2z+1$. Taking $\overline R$ such that $\pi_{\overline R}(C_\Delta^{P,J})=\ev_{P_i}(V_1)$, Proposition \ref{commute} shows that $\pi_{\overline R}(S_\Delta^{P,J})=\ev_{P_i}(V_1)\cap \mathbb{F}_{p^h}^{p^h+1}$ is also MDS. Then, Proposition \ref{method} applied to $S_\Delta^{P,J}$, $\Delta$ being either $\Delta_1$ or $\Delta_2$, proves that $S_\Delta^{P,J}$ is an LRC with locality $(2z+1,p^h-2z+1)$. Replacing $(V_1,a,z,\Omega_z)$ by $(V_2,b,u,\Omega^*_u)$ one deduces that $S_\Delta^{P,J}$ is an LRC with locality $(2u+2,p^h-2u)$, whenever $\Delta$ is either $\Delta_1^*$ or $\Delta_2^*$. Notice that $r$ and $\delta$ do not depend neither on $t$ nor on $v$, unlike dimension and minimum distance.

The case $i=2$ can be proved analogously noticing that we are in the symmetric situation. It suffices to interpolate with respect to $Y$ and change $i$ by $i'$ and $n_2$ by $n_1$.
\medskip

With notation as in Section \ref{work} page \pageref{rect} and $i=1$, we assert that the minimum distance of the code $S_\Delta^{P,J}$ admits the bound on the minimum distance of $C_{\Delta'}^{P,J}$, $d_0\left(C_{\Delta'}^{P,J}\right)$, whenever
$$\left(\Delta,\Delta'\right)\in \Big\{\left(\Delta_1,\Delta^1_{2z,n_2-1}\right),\left(\Delta_2,\Delta^2_{2z,2t}\right),\left(\Delta_1^*,\Delta^1_{2u+1,n_2-1}\right),\left(\Delta_2^*,\Delta^2_{2u+1,2v+1}\right)\Big\}.$$
Let us prove the statement. Figure \ref{fig:proof1} considers the case $(p,h,l,z,t,u,v)=(2,3,6,3,2,2,1)$ to illustrate our reasoning. Let $\mathbf{c}=\ev_P(f)$, $f(X,Y)\in V_\Delta$ be a codeword in $S_\Delta^{P,J}$.

a) Assume firstly that $\Delta=\Delta_1$. A no-root $(\alpha,\beta)$ in $P$ of $f(X,Y)$ must satisfy that $\alpha$ is a no-root of $f(X,\beta)$ as a polynomial in $X$ and $\beta$ is a no-root of $f(\alpha,Y)$ as a polynomial in $Y$. Denote $n_{\beta}$ (respectively, $n_{\alpha}$) the cardinality of the set of no-roots of $f(X,\beta)$ (respectively, $f(\alpha,Y)$). Set $n_X$ (respectively, $n_Y$) the minimum of $n_{\beta}$ (respectively, $n_{\alpha}$) when $\beta$ (respectively, $\alpha$) runs over $P_2$ (respectively, $P_1$). Then, the number of no-roots of $f$ in $P$ is at least $n_Xn_Y$. Since $\dis\left(\ev_{P_1}(V_1)\cap \mathbb{F}_{p^h}^{p^h+1}\right)=p^h+1-2z$ and $\dis\left(\ev_{P_2}\left(\langle Y^e \mid e\in\{0,1,\dots,n_2-1\}\rangle\right)\cap \mathbb{F}_{p^h}^{n_2}\right)=1$ (they are MDS codes), then $\w(\mathbf{c})\geq p^h+1-2z$ and $\dis\left(S_\Delta^{P,J}\right)\geq p^h+1-2z=d_0\left(C_{\Delta'}^{P,J}\right)$. See a) in Figure \ref{fig:proof1}.
    
b) Consider now the case $\Delta=\Delta_1^*$. Since $\dis\left(\ev_{P_1}(V_2)\cap \mathbb{F}_{p^h}^{n_2}\right)=p^h-2u$, the same argument as in a) proves $\dis\left(S_\Delta^{P,J}\right)\geq p^h-2u=d_0\left(C_{\Delta'}^{P,J}\right)$. See b) in Figure \ref{fig:proof1}.
    
c) For proving the case $\Delta=\Delta_2$, we use the following (lexicographical) ordering in $E$:
$$(e_1,e_2)\leq (e'_1,e'_2) \iff e_2 < e'_2 \text{ or (} e_2=e'_2 \text{ and } e_1 \leq e'_1 \text{),}$$
and we distinguish two cases:
\begin{itemize}
    \item  The leading monomial of $f$ is in $\Omega_z\times\{0,1,\dots,n_2-2\}$, then an analogous argument as in a) proves $\w(\mathbf{c})\geq 2(p^h+1-2z)$.

    \item The leading monomial of $f$ is in $\Omega_t\times\{n_2-1\}$, then consider $\Delta'':=\Delta+(t,0) \subseteq E$ because of the relation $p^h+1=0$ in $\{0,1,\dots,p^h\}$. Consider the codeword in $C_{\Delta''}^P$
    $$\ev_P(X^tf)=\ev_P(X^t)*\ev_P(f)=\ev_P(X^t)*\mathbf{c}.$$
    Since $0\notin P_1$, $\w(\mathbf{c})=\w(\ev_P(X^tf))\geq \dis(2t,n_2-1)=p^h+1-2t$ by Proposition \ref{footprint} (the leading monomial of $X^tf$ is $\mu X^\gamma Y^{n_2-1}$ with $\gamma \leq 2t$).
\end{itemize}
Then, $\w(\mathbf{c})\geq \min\{2(p^h+1-2z),p^h+1-2t\}=p^h+1-2t$ and therefore
$$\dis\left(S_\Delta^{P,J}\right)\geq p^h+1-2t=d_0\left(C_{\Delta'}^{P,J}\right).$$
See c) in Figure \ref{fig:proof1}.

d) Finally, when $\Delta=\Delta_2^*$, reasoning as in c) with $\Delta'':=\Delta+(\frac{p^h}{2}+v+1,0)$, one gets the desired bound: 
$$\dis\left(S_\Delta^{P,J}\right)\geq p^h-2v=d_0\left(C_{\Delta'}^{P,J}\right).$$
See d) in Figure \ref{fig:proof1}.

\begin{figure}[b]
    \centering
    \begin{subfigure}[b]{1\textwidth}
    \begin{multicols}{2}
    \centering
    \begin{subfigure}[b]{0.33\textwidth}
        \addtocounter{subfigure}{-1}
        \renewcommand{\thesubfigure}{i}
        \centering
        \begin{tikzpicture}[y=0.5cm, x=0.5cm,font=\normalsize]
        \filldraw[fill=gray!30] (0,0) rectangle (3,4);
        \filldraw[fill=gray!30] (5,0) rectangle (7,4);
        
        \draw (3,0) -- (5,0);
        
        \node [below] at (0,0)  {\scriptsize$0$};
        \node [below] at (1,0) {\scriptsize$1$};
        \node [below] at (2,0) {\scriptsize$2$};
        \node [below] at (3,0) {\scriptsize$3$};
        \node [below] at (4,0) {$\dots$};
        \node [below] at (5,0) {\scriptsize$6$};
        \node [below] at (6,0) {\scriptsize$7$};
        \node [below] at (7,0) {\scriptsize$8$};
        \node [left] at (0,0) {\scriptsize$0$};
        \node [left] at (0,1) {\scriptsize$1$};
        \node [left] at (0,2) {\scriptsize$2$};
        \node [left] at (0,3) {$\vdots$};
        \node [left] at (0,4) {\scriptsize$n_2-1$};
        \end{tikzpicture}
        \caption{$\Delta=\Delta_1=\break \Omega_3\times\{0,1,\dots,n_2-1\}$}
    \end{subfigure}
    
    \columnbreak
    \centering
    \begin{subfigure}[b]{0.33\textwidth}
        \addtocounter{subfigure}{-1}
        \renewcommand\thesubfigure{ii}
        \centering
        \begin{tikzpicture}[y=0.5cm, x=0.5cm,font=\normalsize]
        \filldraw[fill=gray!30] (0,0) rectangle (4,4);
        
        \draw (4,0) -- (6,0);
        
        \node [below] at (0,0)  {\scriptsize$0$};
        \node [below] at (1,0) {\scriptsize$1$};
        \node [below] at (2,0) {\scriptsize$2$};
        \node [below] at (3,0) {$\dots$};
        \node [below] at (4,0) {\scriptsize$6$};
        \node [below] at (5,0) {\scriptsize$7$};
        \node [below] at (6,0) {\scriptsize$8$};
        \node [left] at (0,0) {\scriptsize$0$};
        \node [left] at (0,1) {\scriptsize$1$};
        \node [left] at (0,2) {\scriptsize$2$};
        \node [left] at (0,3) {$\vdots$};
        \node [left] at (0,4) {\scriptsize$n_2-1$};
        
        \end{tikzpicture}
        \caption{$\Delta'=(\Delta_1)'=\Delta_{6,n_2-1}^1$}
    \end{subfigure}
    \end{multicols}
    \caption*{a) Either $P_2=U_{n_2}$, $n_2\mid q-1$ and $J=\{1,2\}$, or $P_2=U_{n_2-1}\cup\{0\}$, $n_2-1\mid q-1$ and $J=\{1\}$}
    \end{subfigure}

    \centering
    \hfill\break
    \begin{subfigure}[b]{1\textwidth}
    \begin{multicols}{2}
    \centering
    \begin{subfigure}[b]{0.33\textwidth}
        \addtocounter{subfigure}{-1}
        \renewcommand\thesubfigure{i}
        \centering
        \begin{tikzpicture}[y=0.5cm, x=0.5cm,font=\normalsize]
        \filldraw[fill=gray!30] (2,0) rectangle (5,4);
        
        \draw (0,0) -- (0,4);
        \draw (0,0) -- (2,0);
        \draw (5,0) -- (6,0);
        
        \node [below] at (0,0)  {\scriptsize$0$};
        \node [below] at (1,0) {\scriptsize$1$};
        \node [below] at (2,0) {\scriptsize$2$};
        \node [below] at (3,0) {\scriptsize$3$};
        \node [below] at (4,0) {$\dots$};
        \node [below] at (5,0) {\scriptsize$7$};
        \node [below] at (6,0) {\scriptsize$8$};
        \node [left] at (0,0) {\scriptsize$0$};
        \node [left] at (0,1) {\scriptsize$1$};
        \node [left] at (0,2) {\scriptsize$2$};
        \node [left] at (0,3) {$\vdots$};
        \node [left] at (0,4) {\scriptsize$n_2-1$};
        \end{tikzpicture}
        \caption{$\Delta=\Delta_1^*= \break \Omega^*_2\times\{0,1,\dots,n_2-1\}$}
    \end{subfigure}
    
    \columnbreak
    \centering
    \begin{subfigure}[b]{0.33\textwidth}
        \addtocounter{subfigure}{-1}
        \renewcommand\thesubfigure{ii}
        \centering
        \begin{tikzpicture}[y=0.5cm, x=0.5cm,font=\normalsize]
        \filldraw[fill=gray!30] (0,0) rectangle (4,4);
        
        \draw (4,0) -- (6,0);
        
        \node [below] at (0,0)  {\scriptsize$0$};
        \node [below] at (1,0) {\scriptsize$1$};
        \node [below] at (2,0) {\scriptsize$2$};
        \node [below] at (3,0) {$\dots$};
        \node [below] at (4,0) {\scriptsize$5$};
        \node [below] at (5,0) {$\dots$};
        \node [below] at (6,0) {\scriptsize$8$};
        \node [left] at (0,0) {\scriptsize$0$};
        \node [left] at (0,1) {\scriptsize$1$};
        \node [left] at (0,2) {\scriptsize$2$};
        \node [left] at (0,3) {$\vdots$};
        \node [left] at (0,4) {\scriptsize$n_2-1$};
        
        \end{tikzpicture}
        \caption{$\Delta'=(\Delta_1^*)'=\Delta_{5,n_2-1}^1$}
    \end{subfigure}
    \end{multicols}
    \caption*{b) Either $P_2=U_{n_2}$, $n_2\mid q-1$ and $J=\{1,2\}$, or $P_2=U_{n_2-1}\cup\{0\}$, $n_2-1\mid q-1$ and $J=\{1\}$}
    \end{subfigure}
    
    \centering
    \hfill\break
    \begin{subfigure}[b]{1\textwidth}
    \begin{multicols}{3}
    \centering
    \begin{subfigure}[b]{0.33\textwidth}
        \addtocounter{subfigure}{-1}
        \renewcommand\thesubfigure{i}
        \centering
        \begin{tikzpicture}[y=0.5cm, x=0.5cm,font=\normalsize]
        \fill[gray!30] (0,0) rectangle (3,4);
        \fill[gray!30] (0,4) rectangle (2,5);
        \fill[gray!30] (5,0) rectangle (7,4);
        \fill[gray!30] (6,4) rectangle (7,5);
        
        \draw (0,0) -- (7,0);
        \draw (0,0) -- (0,5) -- (2,5) -- (2,4) -- (3,4) -- (3,0);
        \draw (5,0) -- (5,4) -- (6,4) -- (6,5) -- (7,5) -- (7,0);
        
        \node [below] at (0,0)  {\scriptsize$0$};
        \node [below] at (1,0) {\scriptsize$1$};
        \node [below] at (2,0) {\scriptsize$2$};
        \node [below] at (3,0) {\scriptsize$3$};
        \node [below] at (4,0) {$\dots$};
        \node [below] at (5,0) {\scriptsize$6$};
        \node [below] at (6,0) {\scriptsize$7$};
        \node [below] at (7,0) {\scriptsize$8$};
        \node [left] at (0,0) {\scriptsize$0$};
        \node [left] at (0,1) {\scriptsize$1$};
        \node [left] at (0,2) {\scriptsize$2$};
        \node [left] at (0,3) {$\vdots$};
        \node [left] at (0,4) {\scriptsize$n_2-2$};
        \node [left] at (0,5) {\scriptsize$n_2-1$};
        \end{tikzpicture}
        \caption{$\Delta=\Delta_2 = \break \Omega_3\times\{0,1,\dots,n_2-2\} \cup \Omega_2\times \{n_2-1\}$}
    \end{subfigure}
    
    \columnbreak
    \centering
    \begin{subfigure}[b]{0.33\textwidth}
        \addtocounter{subfigure}{-1}
        \renewcommand\thesubfigure{ii}
        \centering
        \begin{tikzpicture}[y=0.5cm, x=0.5cm,font=\normalsize]
        \fill[gray!30] (0,0) rectangle (3,5);
        \fill[gray!30] (3,0) rectangle (4,4);
        
        \draw (0,0) -- (6,0);
        \draw (0,0) -- (0,5) -- (3,5) -- (3,4) -- (4,4) -- (4,0);

        \filldraw[fill=black!40,draw=black!80] (6,0) circle (2pt);
        \filldraw[fill=black!40,draw=black!80] (6,1) circle (2pt);
        \filldraw[fill=black!40,draw=black!80] (6,2) circle (2pt);
        \filldraw[fill=black!40,draw=black!80] (6,3) circle (2pt);
        \filldraw[fill=black!40,draw=black!80] (6,4) circle (2pt);
        
        \node [below] at (0,0)  {\scriptsize$0$};
        \node [below] at (1,0) {\scriptsize$1$};
        \node [below] at (2,0) {$\dots$};
        \node [below] at (3,0) {\scriptsize$4$};
        \node [below] at (4,0) {\scriptsize$5$};
        \node [below] at (5,0) {$\dots$};
        \node [below] at (6,0) {\scriptsize$8$};
        \node [left] at (0,0) {\scriptsize$0$};
        \node [left] at (0,1) {\scriptsize$1$};
        \node [left] at (0,2) {\scriptsize$2$};
        \node [left] at (0,3) {$\vdots$};
        \node [left] at (0,4) {\scriptsize$n_2-2$};
        \node [left] at (0,5) {\scriptsize$n_2-1$};
        
        \end{tikzpicture}
        \caption{$\Delta''=(\Delta_2)''=\Delta_2+(2,0)$}
    \end{subfigure}

        \columnbreak
    \centering
    \begin{subfigure}[b]{0.33\textwidth}
        \addtocounter{subfigure}{-1}
        \renewcommand\thesubfigure{iii}
        \centering
        \begin{tikzpicture}[y=0.5cm, x=0.5cm,font=\normalsize]
        \fill[gray!30] (0,0) rectangle (3,5);
        \fill[gray!30] (3,0) rectangle (5,4);
        
        \draw (0,0) -- (7,0);
        \draw (0,0) -- (0,5) -- (3,5) -- (3,4) -- (5,4) -- (5,0);
        
        \node [below] at (0,0)  {\scriptsize$0$};
        \node [below] at (1,0) {\scriptsize$1$};
        \node [below] at (2,0) {$\dots$};
        \node [below] at (3,0) {\scriptsize$4$};
        \node [below] at (4,0) {\scriptsize$5$};
        \node [below] at (5,0) {\scriptsize$6$};
        \node [below] at (6,0) {\scriptsize$7$};
        \node [below] at (7,0) {\scriptsize$8$};
        \node [left] at (0,0) {\scriptsize$0$};
        \node [left] at (0,1) {\scriptsize$1$};
        \node [left] at (0,2) {\scriptsize$2$};
        \node [left] at (0,3) {$\vdots$};
        \node [left] at (0,4) {\scriptsize$n_2-2$};
        \node [left] at (0,5) {\scriptsize$n_2-1$};
        
        \end{tikzpicture}
        \caption{$\Delta'=(\Delta_2)'=\Delta_{6,4}^2$}
    \end{subfigure}
    \end{multicols}
    \caption*{c) $P_2=U_{n_2-1}\cup\{0\}$, $n_2-1\mid q-1$ and $J=\{1\}$}
    \end{subfigure}

    \centering
    \hfill\break
    \begin{subfigure}[b]{1\textwidth}
    \begin{multicols}{3}
    \centering
        \begin{subfigure}[b]{0.33\textwidth}
        \addtocounter{subfigure}{-1}
        \renewcommand\thesubfigure{i}
        \centering
        \begin{tikzpicture}[y=0.5cm, x=0.5cm,font=\normalsize]
        \fill[gray!30] (2,0) rectangle (6,4);
        \fill[gray!30] (3,4) rectangle (5,5);
        
        \draw (0,0) -- (7,0);
        \draw (0,0) -- (0,5);
        \draw (2,0) -- (2,4) -- (3,4) -- (3,5) -- (5,5) -- (5,4) -- (6,4) -- (6,0);
        
        \node [below] at (0,0) {\scriptsize$0$};
        \node [below] at (1,0) {\scriptsize$1$};
        \node [below] at (2,0) {\scriptsize$2$};
        \node [below] at (3,0) {\scriptsize$3$};
        \node [below] at (4,0) {$\dots$};
        \node [below] at (5,0) {\scriptsize$6$};
        \node [below] at (6,0) {\scriptsize$7$};
        \node [below] at (7,0) {\scriptsize$8$};
        \node [left] at (0,0) {\scriptsize$0$};
        \node [left] at (0,1) {\scriptsize$1$};
        \node [left] at (0,2) {\scriptsize$2$};
        \node [left] at (0,3) {$\vdots$};
        \node [left] at (0,4) {\scriptsize$n_2-2$};
        \node [left] at (0,5) {\scriptsize$n_2-1$};
        
        \end{tikzpicture}
        \caption{$\Delta=\Delta_2^*= \break \Omega_2^*\times\{0,1,\dots,n_2-2\} \cup \Omega_1^*\times \{n_2-1\}$}
    \end{subfigure}

    \columnbreak
    \centering
        \begin{subfigure}[b]{0.33\textwidth}
        \addtocounter{subfigure}{-1}
        \renewcommand\thesubfigure{ii}
        \centering
        \begin{tikzpicture}[y=0.5cm, x=0.5cm,font=\normalsize]
        \fill[gray!30] (0,0) rectangle (3,5);
        \fill[gray!30] (3,0) rectangle (4,4);
        
        \draw (0,0) -- (6,0);
        \draw (0,0) -- (0,5) -- (3,5) -- (3,4) -- (4,4) -- (4,0);
        
        \filldraw[fill=black!40,draw=black!80] (6,0) circle (2pt);
        \filldraw[fill=black!40,draw=black!80] (6,1) circle (2pt);
        \filldraw[fill=black!40,draw=black!80] (6,2) circle (2pt);
        \filldraw[fill=black!40,draw=black!80] (6,3) circle (2pt);
        \filldraw[fill=black!40,draw=black!80] (6,4) circle (2pt);
        
        \node [below] at (0,0) {\scriptsize$0$};
        \node [below] at (1,0) {\scriptsize$1$};
        \node [below] at (2,0) {\scriptsize$2$};
        \node [below] at (3,0) {\scriptsize$3$};
        \node [below] at (4,0) {\scriptsize$4$};
        \node [below] at (5,0) {$\dots$};
        \node [below] at (6,0) {\scriptsize$8$};
        \node [left] at (0,0) {\scriptsize$0$};
        \node [left] at (0,1) {\scriptsize$1$};
        \node [left] at (0,2) {\scriptsize$2$};
        \node [left] at (0,3) {$\vdots$};
        \node [left] at (0,4) {\scriptsize$n_2-2$};
        \node [left] at (0,5) {\scriptsize$n_2-1$};
        
        \end{tikzpicture}
        \caption{$\Delta''=(\Delta_2^*)''=\Delta_2^*+(6,0)$}
    \end{subfigure}

    \columnbreak
    \centering
        \begin{subfigure}[b]{0.33\textwidth}
        \addtocounter{subfigure}{-1}
        \renewcommand\thesubfigure{iii}
        \centering
        \begin{tikzpicture}[y=0.5cm, x=0.5cm,font=\normalsize]
        \fill[gray!30] (0,0) rectangle (3,5);
        \fill[gray!30] (3,0) rectangle (5,4);
        
        \draw (0,0) -- (7,0);
        \draw (0,0) -- (0,5) -- (3,5) -- (3,4) -- (5,4) -- (5,0);
        
        \node [below] at (0,0) {\scriptsize$0$};
        \node [below] at (1,0) {\scriptsize$1$};
        \node [below] at (2,0) {\scriptsize$2$};
        \node [below] at (3,0) {\scriptsize$3$};
        \node [below] at (4,0) {\scriptsize$4$};
        \node [below] at (5,0) {\scriptsize$5$};
        \node [below] at (6,0) {$\dots$};
        \node [below] at (7,0) {\scriptsize$8$};
        \node [left] at (0,0) {\scriptsize$0$};
        \node [left] at (0,1) {\scriptsize$1$};
        \node [left] at (0,2) {\scriptsize$2$};
        \node [left] at (0,3) {$\vdots$};
        \node [left] at (0,4) {\scriptsize$n_2-2$};
        \node [left] at (0,5) {\scriptsize$n_2-1$};
        
        \end{tikzpicture}
        \caption{$\Delta'=(\Delta_2^*)'=\Delta_{5,3}^2$}
    \end{subfigure}

    \end{multicols}
    \caption*{d) $P_2=U_{n_2-1}\cup\{0\}$, $n_2-1\mid q-1$ and $J=\{1\}$}
    \end{subfigure}
    
    \caption{Sets $\Delta$, $\Delta'$ (and $\Delta''$) considered in the proof of Proposition \ref{bivq+1} for values $(i,p^h,q,P_1,z,t,u,v)=(1,8,64,U_9,3,2,2,1)$}
    \label{fig:proof1}
\end{figure}

The case $i=2$ follows by symmetry. It suffices to replace $P_1$ by $P_2$, $n_2$ by $n_1$ and consider
$$\left(\Delta,\Delta'\right)\in \Bigg\{\left(\Delta_1,\Delta^1_{n_1-1,2z}\right),\left(\Delta_2,\Delta^{2,\sigma}_{2z,2t}\right),\left(\Delta_1^*,\Delta^1_{n_1-1,2u+1}\right),\left(\Delta_2^*,\Delta^{2,\sigma}_{2u+1,2v+1}\right)\Bigg\}.$$
\medskip

Notice that $\#\Delta=\#\Delta'$ and $\dim(S_\Delta^{P,J})=\dim(C_\Delta^{P,J})=\dim(C_{\Delta'}^{P,J})$. Moreover, $\dis(S_\Delta^{P,J})\geq \dis(C_\Delta^{P,J})\geq d_0(C_{\Delta'}^{P,J})$ and the locality of $C_{\Delta'}^{P,J}$ is the same as the locality of $S_\Delta^{P,J}$. Then, the fact that $C_{\Delta'}^{P,J}$ is optimal (Corollary \ref{colpar}(1) when $\Delta$ is $\Delta_1$ or $\Delta_1^*$, and Corollary \ref{colpar}(2) when $\Delta$ is $\Delta_2$ or $\Delta_2^*$) implies that the subfield-subcode $S_\Delta^{P,J}$ over the field $\mathbb{F}_{p^h}$ is optimal, which concludes the proof.

\end{proof}

At the beginning of this subsection we announced the introduction of two families of new optimal codes. We start by giving some sets that will be useful for introducing our second family. In this case, $p=2$, $l\geq 4$ is an even positive integer, $h=\frac{l}{2}$ and $P_i=U_{2^h+1}\cup\{0\}\subseteq \mathbb{F}_q$. Recall that $\{i,i'\}=\{1,2\}$. Then, $n_i = 2^h + 2$ and $P = P_1 \times P_2$, where $P_{i'}$ is either $U_{n_{i'}}\subseteq \mathbb{F}_q$ , with $n_{i'}\mid q-1$ and $J = \{i'\}$, or $U_{n_{i'}-1}\cup\{0\}\subseteq \mathbb{F}_q$, with $n_{i'}- 1 \mid q- 1$ and $J = \emptyset$.

Now we introduce some sets which will be the sets $\supp_{X_i}(V_\Delta)$ corresponding to the sets $\Delta$ that we are going to consider. Let $1\leq j\leq n_{i'}-1$ and $2\leq z\leq 3$, $2^h-2z+1\geq \max\left\{0,2^h-6\right\}$ be positive integers and denote
        $$\Omega:=\left\{0,1,2^h\right\}=\Lambda^i_0\cup\Lambda^i_1,$$
        $$\Omega^\perp:=\left\{0,2,3,\dots,2^h-1\right\}=\left\{0,1,\dots,2^h+1\right\}\big\backslash\left(\Lambda^i_1\cup\Lambda^i_{2^h+1}\right)$$
    and
        $$\Omega^*(z)=\Omega^*:=\left\{z,z+1,\dots,2^h-z+1\right\}=\Lambda^i_z\cup\Lambda^i_{z+1}\cup\cdots\cup\Lambda^i_{2^h-1},$$
which are closed sets of $\{0,1,\dots,n_i-1\}=\{0,1,\dots,2^h+1\}$ (identified with $\{0\}\cup\mathbb{Z}/(2^h+1)\mathbb{Z}$) in the variable $i$ with respect to $2^h$. Define
        $$\Delta_1:=\begin{cases}
            \Omega\times\{0,1,\dots,n_2-1\}, & \text{when } i=1,\\
            \{0,1,\dots,n_1-1\}\times \Omega, & \text{otherwise;}
        \end{cases}$$
        $$\Delta^\perp_1:=\begin{cases}
            \Omega^\perp\times\{0,1,\dots,n_2-1\}, & \text{when } i=1,\\
            \{0,1,\dots,n_1-1\}\times \Omega^\perp, & \text{otherwise;}
        \end{cases}
        $$
        $$\Delta_2(j)=\Delta_2:=\begin{cases}
            \Omega\times\{0,1,\dots,j-1\} \cup (0,j), & \text{when } i=1,\\
            \{0,1,\dots,j-1\}\times \Omega \cup (j,0), & \text{otherwise;}
        \end{cases}
        $$
        
and
        $$\Delta^\perp_2(z)=\Delta^\perp_2:=\begin{cases}
            \Omega^\perp\times\{0,1,\dots,n_2-2\} \cup \Omega^*\times\{n_2-1\}, & \text{when } i=1,\\
            \{0,1,\dots,n_1-2\}\times \Omega^\perp \cup \{n_1-1\} \times \Omega^*, & \text{otherwise.}
        \end{cases}
        $$

Our next result plays the role of Lemma \ref{lemmaph+1} for studying our second family of optimal codes.

\begin{lemma}\label{lemmaph+2}
Keep the above notation. Let $V_1=\langle X^e \mid e \in \Omega \rangle _{\mathbb{F}_q}$, $V_2=\langle X^e \mid e \in \Omega^\perp \rangle _{\mathbb{F}_q} \subseteq \mathcal{R}_i$ and define $C_1:=\ev_{P_i}(V_1) \cap \mathbb{F}_{2^h}^{\#P_i=2^h+2}$ and $C_2:=\ev_{P_i}(V_2) \cap \mathbb{F}_{2^h}^{2^h+2}$. Then, $C_1$ and $C_2$ are MDS codes.
\end{lemma}

\begin{proof}
Notice that $\ev_{P_i}(V_2)$ is the dual code of $\ev_{P_i}(V_1)$ since
$$\Omega^\perp=\left\{0,1,\dots,2^h+1\right\}\big\backslash \left\{2^h+1-x \mid x\in \Omega\right\}$$
\cite[Proposition 1]{GHR2015} and, by Delsarte Theorem, $(C_2)^\perp=\left(\ev_{P_i}(V_2)\right)^\perp \cap \mathbb{F}_{2^h}^{2^h+2}=C_1$. Thus, it suffices to prove that $C_1$ is an MDS code. Notice that its dimension coincides with the dimension of $\ev_{P_i}(V_1)$ because $\Omega=\Lambda_0^i\cup\Lambda_1^i$ is closed \cite[Theorem 2.3]{GHM2020}, so the parameters of $C_1$ are $[2^h+2,3,\leq 2^h]_{2^h}$. Moreover, any codeword $\mathbf{c}\in C_1$ is of the form $\mathbf{c}=\ev_{P_i}(f)$, where $f=\mathcal{T}(\lambda+\mu X)$, $\lambda$, $\mu\in \mathbb{F}_{q}=\mathbb{F}_{2^{2h}}$ and $\mathcal{T}\colon \mathcal{R}_i \to \mathcal{R}_i$ is the map given by $\mathcal{T}(g)=g+g^{2^h}$ \cite[Proposition 5]{GH2015}. We have to prove that $\dis(C_1)=2^h$, which is equivalent to prove that the number of roots of $f=\lambda+\lambda^{2^h}+\mu X+\mu^{2^h}X^{2^h}$ in $P_i=U_{2^h+1}\cup\{0\}$ is at most 2, or that the equation
\begin{equation}\label{eqroots}
\lambda+\mu X=\lambda^{2^h}+\mu^{2^h}X^{2^h}
\end{equation}
has at most 2 solutions in $P_i$. Indeed, if $\lambda\notin \mathbb{F}_{2^h}$, $X=0$ is not a solution since $\lambda\neq \lambda^{2^h}$. Thus, the above equation is equivalent to
$$\lambda X+\mu X^2=\lambda^{2^h}X+\mu^{2^h}X^{2^h+1}$$
and to
$$\mu X^2+\left(\lambda+\lambda^{2^h}\right)X+\mu^{2^h}=0,$$
which has at most 2 solutions in $P_i$. Otherwise, if $\lambda\in \mathbb{F}_{2^h}$, then $\lambda= \lambda^{2^h}$ and (\ref{eqroots}) is equivalent to
$$\mu X\left(\left(\mu X\right)^{2^h-1}-1\right)=0.$$
We may suppose $\mu\neq 0$ since the case $\mu=0$ is not relevant to compute the minimum distance. Then, the solutions are $X=0$ and $X=\frac{\beta}{\mu}$ with $\beta\in\mathbb{F}_{2^h}$ such that $\beta^{2^h+1}=\mu^{2^h+1}$ (since $X^{2^h+1}=1$), that is, $\beta^2=\mu^{2^h+1}$. The solution $X=\frac{\beta}{\mu}$ exists if $\mu^{2^h+1}$ ($\in\mathbb{F}_{2^h}$) is a square in $\in\mathbb{F}_{2^h}$ and therefore $\beta=\sqrt{\mu^{2^h+1}}$. Hence, we obtain at most 2 solutions in $P_i$, as desired.
\end{proof}

\begin{prop}\label{bivq+2}
    Keep the notation as before Lemma \ref{lemmaph+2} where $\;\mathbb{F}_{2^h}$ is regarded as a subfield of $\;\mathbb{F}_{q=2^{2h}}$. Fixed $i\in\{1,2\}$ and $P_i=U_{2^h+1}\cup\{0\}$, the set of $2^h+1$-th roots of unity together with $0$, the following statements determine sets $P_{i'}$, $J$ and $\Delta$ such that the subfield-subcodes $S_\Delta^{P,J}$ over the field $\;\mathbb{F}_{2^h}$ are optimal $(r,\delta)$-LRCs. Recall that $P=P_1\times P_2$ and $\{i,i'\}=\{1,2\}$.
\begin{enumerate}
    \item $P_{i'}=U_{n_{i'}}$ for some $n_{i'}$ such that $n_{i'}\mid q-1$; $J=\{i'\}$ and $\Delta=\Delta_1$, in which case $(r,\delta)=(3,2^h)$.

    \item $P_{i'}=U_{n_{i'}-1}\cup\{0\}$ for some $n_{i'}$ such that $n_{i'}-1\mid q-1$; $J=\emptyset$ and $\Delta=\Delta_1$, in which case $(r,\delta)=(3,2^h)$.

    \item $P_{i'}=U_{n_{i'}}$ for some $n_{i'}$ such that $n_{i'}\mid q-1$; $J=\{i'\}$ and $\Delta=\Delta_1^\perp$, in which case $(r,\delta)=(2^h-1,4)$.

    \item $P_{i'}=U_{n_{i'}-1}\cup\{0\}$ for some $n_{i'}$ such that $n_{i'}-1\mid q-1$; $J=\emptyset$ and $\Delta=\Delta_1^\perp$, in which case $(r,\delta)=(2^h-1,4)$.

    \item $P_{i'}=U_{n_{i'}}$ for some $n_{i'}$ such that $n_{i'}\mid 2^h-1$; $J=\{i'\}$ and $\Delta=\Delta_2$, where $j\geq \max\{1,n_{i'}-2^{h-1}\}$. In this case $(r,\delta)=(3,2^h)$.

    \item $P_{i'}=U_{n_{i'}-1}\cup\{0\}$ for some $n_{i'}$ such that $n_{i'}-1\mid 2^h-1$; $J=\emptyset$ and $\Delta=\Delta_2$, where $\max\{1,n_{i'}-2^{h-1}\}\leq j<n_{i'}-1$. In this case $(r,\delta)=(3,2^h)$.

    \item $P_{i'}=U_{n_{i'}-1}\cup\{0\}$ for some $n_{i'}$ such that $n_{i'}-1\mid q-1$; $J=\emptyset$ and $\Delta=\Delta_2$, where $j=n_{i'}-1$. In this case $(r,\delta)=(3,2^h)$.

    \item $P_{i'}=U_{n_{i'}-1}\cup\{0\}$ for some $n_{i'}$ such that $n_{i'}-1\mid q-1$; $J=\emptyset$ and $\Delta=\Delta_2^\perp$, in which case $(r,\delta)=(2^h-1,4)$.

\end{enumerate}
\end{prop}

\begin{proof}
    The proof follows from a close reasoning to that given in the proof of Proposition \ref{bivq+1}. There are some minor differences which we next explain.
    
        - Recall that $\{0,1,\dots,2^h+1\}$ is a set of representatives of $\{0\}\cup\mathbb{Z}/(2^h+1)\mathbb{Z}$ and $\Lambda^i_l$ is the minimal closed set in the variable $i$ of the element $l\in \{0,1,\dots,2^h+1\}$. Then, as we said before
        $$\Omega=\Lambda^i_0\cup\Lambda^i_1,$$
        $$\Omega^\perp=\left\{0,1,\dots,2^h+1\right\}\big\backslash\left(\Lambda^i_1\cup\Lambda^i_{2^h+1}\right),$$
        and
        $$\Omega^*=\Lambda^i_z\cup\Lambda^i_{z+1}\cup\cdots\cup\Lambda^i_{2^h-1},$$
        are clearly closed sets, which proves that the sets $\Delta_1$, $\Delta_1^\perp$ and $\Delta_2^\perp$, as well as $\Delta_2$ in item (7) are closed. The fact that the sets $\Delta_2$ in items (5) and (6) are closed follows by noticing that when $P_{i'}=U_{n_{i'}}$, $n_{i'}\mid 2^h-1$ or $P_{i'}=U_{n_{i'}-1}\cup\{0\}$, $n_{i'}-1\mid 2^h-1$, one can identify $2^h$ with 1 when computing minimal closed sets in the variable $i'$. Therefore, the sets $\{0,1,\dots,j-1\}$ and $\{j\}$ are closed because they are a union of single point minimal closed sets. This proves that $\Delta_2$ is closed.
 
        - Lemma \ref{lemmaph+2} implies that $\ev_{P_i}(V_1) \cap \mathbb{F}_{2^h}^{2^h+2}$ and $\ev_{P_i}(V_2) \cap \mathbb{F}_{2^h}^{2^h+2}$ are MDS codes with respective minimum distances $2^h$ and $4$. Proposition \ref{method} applied to $S_\Delta^{P,J}$ proves that it is an LRC with locality $(3,2^h)$ when $\Delta$ equals $\Delta_1$ or $\Delta_2$ and $(2^h-1,4)$ in case $\Delta$ be $\Delta_1^\perp$ or $\Delta_2^\perp$.

        - When $i=1$, the minimum distance of $S_\Delta^{P,J}$ admits the bound on the minimum distance of $C_{\Delta'}^{P,J}$, $d_0\left(C_{\Delta'}^{P,J}\right)$, whenever the pair $\left(\Delta,\Delta'\right)$ belongs to the following set:
        $$\left\{\left(\Delta_1,\Delta^1_{2,n_2-1}\right),\left(\Delta_1^\perp,\Delta^1_{2^h-2,n_2-1}\right),\left(\Delta_2,\begin{cases}
            \Delta^2_{2,0}, & \text{when } j=n_2-1,\\
            \Delta^3_{2,j}, & \text{otherwise.}
        \end{cases}\right),\left(\Delta_2^\perp,\Delta^2_{2^h-2,2^h-2z+1}\right)\right\}.$$
        Recall that the sets $\Delta^l_{i,j}$, $1\leq l\leq 3$ were introduced in Section \ref{work}. The cases where $\Delta$ equals $\Delta_1$ or $\Delta_1^\perp$ (respectively, $\Delta_2$ or $\Delta_2^\perp$) can be proved as in item a) (respectively c)) in the proof of Proposition \ref{bivq+1}. However, when $\Delta=\Delta_2$ and the exponent of the leading monomial of $f$ is $(0,j)$, we do not consider any set $\Delta''$ but we immediately notice that $\w(\mathbf{c})\geq n_1(n_2-j)$. When $\Delta=\Delta_2^\perp$ and the exponent of the leading monomial of $f$ is in $\Omega^*\times\{n_2-1\}$, following the idea of the proof of Proposition \ref{bivq+1}, we consider the sets:
        $$\Delta''_0:=\Delta+(2^h+2-z,0) \subseteq E \quad \textrm{and} \quad \Delta'':=\Delta''_0+(-1,0)\subseteq E,$$
        because of the relation $2^h+2=1$ in $\{0,1,\dots,2^h+1\}$. We illustrate this part of the proof with the example in Figure \ref{fig:proof2}. Since $0\in P_1$, now we have $\w(\mathbf{c})\geq \w(\ev_P(X^{-1}(X^{2^h+2-z}f)))  \geq \dis(2^h+1-2z,n_2-1)=2z+1$. Then, wherever the exponent of the leading monomial of $f$ is, $\w(\mathbf{c})\geq \min\{8,2z+1\}=2z+1$ and therefore the minimum distance of $S_\Delta^{P,J}$ admits the bound on the minimum distance of $C_{\Delta'}^{P,J}$, that is, $\dis(S_\Delta^{P,J})\geq 2z+1=d_0(C_{\Delta'}^{P,J})$.\hfill\break
        \begin{figure}[h]
    \centering
    \begin{multicols}{4}
    \centering
    \begin{subfigure}[b]{0.2\textwidth}
        \centering
        \begin{tikzpicture}[y=0.5cm, x=0.5cm,font=\normalsize]
        \filldraw[fill=gray!30] (2,0) rectangle (3,5);
        
        \draw (0,0) -- (0,5);
        \draw (0,0) -- (2,0);
        \draw (3,0) -- (5,0);
        
        \filldraw[fill=black!40,draw=black!80] (0,0) circle (2pt);
        \filldraw[fill=black!40,draw=black!80] (0,1) circle (2pt);
        \filldraw[fill=black!40,draw=black!80] (0,2) circle (2pt);
        \filldraw[fill=black!40,draw=black!80] (0,3) circle (2pt);
        \filldraw[fill=black!40,draw=black!80] (0,4) circle (2pt);
        
        \node [below] at (0,0)  {\scriptsize$0$};
        \node [below] at (1,0) {\scriptsize$1$};
        \node [below] at (2,0) {\scriptsize$2$};
        \node [below] at (3,0) {\scriptsize$3$};
        \node [below] at (4,0) {\scriptsize$4$};
        \node [below] at (5,0) {\scriptsize$5$};
        \node [left] at (0,0) {\scriptsize$0$};
        \node [left] at (0,1) {\scriptsize$1$};
        \node [left] at (0,2) {\scriptsize$2$};
        \node [left] at (0,3) {$\vdots$};
        \node [left] at (0,4) {\scriptsize$n_2-2$};
        \node [left] at (0,5) {\scriptsize$n_2-1$};
        \end{tikzpicture}
        \caption{$\Delta=\Delta_2^\perp \break = \Omega^\perp\times \{0,1,\dots,n_2-2\} \cup \Omega^*\times \{n_2-1\}$}
    \end{subfigure}

    \columnbreak
    \centering
    \begin{subfigure}[b]{0.2\textwidth}
        \centering
        \begin{tikzpicture}[y=0.5cm, x=0.5cm,font=\normalsize]
        \filldraw[fill=gray!30] (1,0) rectangle (2,5);
        
        \draw (0,0) -- (0,5);
        \draw (0,0) -- (1,0);
        \draw (2,0) -- (5,0);
        
        \filldraw[fill=black!40,draw=black!80] (4,0) circle (2pt);
        \filldraw[fill=black!40,draw=black!80] (4,1) circle (2pt);
        \filldraw[fill=black!40,draw=black!80] (4,2) circle (2pt);
        \filldraw[fill=black!40,draw=black!80] (4,3) circle (2pt);
        \filldraw[fill=black!40,draw=black!80] (4,4) circle (2pt);
        
        \node [below] at (0,0)  {\scriptsize$0$};
        \node [below] at (1,0) {\scriptsize$1$};
        \node [below] at (2,0) {\scriptsize$2$};
        \node [below] at (3,0) {\scriptsize$3$};
        \node [below] at (4,0) {\scriptsize$4$};
        \node [below] at (5,0) {\scriptsize$5$};
        \node [left] at (0,0) {\scriptsize$0$};
        \node [left] at (0,1) {\scriptsize$1$};
        \node [left] at (0,2) {\scriptsize$2$};
        \node [left] at (0,3) {$\vdots$};
        \node [left] at (0,4) {\scriptsize$n_2-2$};
        \node [left] at (0,5) {\scriptsize$n_2-1$};
        
        \end{tikzpicture}
        \caption{$\Delta_0''=(\Delta_2^\perp)_0'' \break =\Delta_2^\perp+(4,0)$}
    \end{subfigure}

    \columnbreak
    \centering
    \begin{subfigure}[b]{0.2\textwidth}
        \centering
        \begin{tikzpicture}[y=0.5cm, x=0.5cm,font=\normalsize]
        \filldraw[fill=gray!30] (0,0) rectangle (1,5);
        
        \draw (1,0) -- (5,0);
        
        \filldraw[fill=black!40,draw=black!80] (3,0) circle (2pt);
        \filldraw[fill=black!40,draw=black!80] (3,1) circle (2pt);
        \filldraw[fill=black!40,draw=black!80] (3,2) circle (2pt);
        \filldraw[fill=black!40,draw=black!80] (3,3) circle (2pt);
        \filldraw[fill=black!40,draw=black!80] (3,4) circle (2pt);
        
        \node [below] at (0,0)  {\scriptsize$0$};
        \node [below] at (1,0) {\scriptsize$1$};
        \node [below] at (2,0) {\scriptsize$2$};
        \node [below] at (3,0) {\scriptsize$3$};
        \node [below] at (4,0) {\scriptsize$4$};
        \node [below] at (5,0) {\scriptsize$5$};
        \node [left] at (0,0) {\scriptsize$0$};
        \node [left] at (0,1) {\scriptsize$1$};
        \node [left] at (0,2) {\scriptsize$2$};
        \node [left] at (0,3) {$\vdots$};
        \node [left] at (0,4) {\scriptsize$n_2-2$};
        \node [left] at (0,5) {\scriptsize$n_2-1$};
        
        \end{tikzpicture}
        \caption{$\Delta''=(\Delta_2^\perp)'' \break =(\Delta_2^\perp)_0''+(-1,0)$}
    \end{subfigure}

    \columnbreak
    \centering
    \begin{subfigure}[b]{0.2\textwidth}
        \centering
        \begin{tikzpicture}[y=0.5cm, x=0.5cm,font=\normalsize]
        \fill[gray!30] (0,0) rectangle (1,5);
        \fill[gray!30] (0,0) rectangle (2,4);
    
        \draw (0,0) -- (5,0);
        \draw (0,0) --(0,5) -- (1,5) -- (1,4) -- (2,4) -- (2,0);
        
        \node [below] at (0,0)  {\scriptsize$0$};
        \node [below] at (1,0) {\scriptsize$1$};
        \node [below] at (2,0) {\scriptsize$2$};
        \node [below] at (3,0) {\scriptsize$3$};
        \node [below] at (4,0) {\scriptsize$4$};
        \node [below] at (5,0) {\scriptsize$5$};
        \node [left] at (0,0) {\scriptsize$0$};
        \node [left] at (0,1) {\scriptsize$1$};
        \node [left] at (0,2) {\scriptsize$2$};
        \node [left] at (0,3) {$\vdots$};
        \node [left] at (0,4) {\scriptsize$n_2-2$};
        \node [left] at (0,5) {\scriptsize$n_2-1$};
        
        \end{tikzpicture}
        \caption{$\Delta'=(\Delta_2^\perp)'=\Delta_{2,1}^2$}
    \end{subfigure}
    \end{multicols}
    \caption{Sets $\Delta_2^\perp$, $\Delta''_0$, $\Delta''$ and $\Delta'$ considered in the proof of Proposition \ref{bivq+2} for values $(i,2^h,q,P_1,P_2,J,z)=(1,4,16,U_5\cup\{0\},U_{n_2-1}\cup\{0\},\emptyset,2)$}
    \label{fig:proof2}
\end{figure}
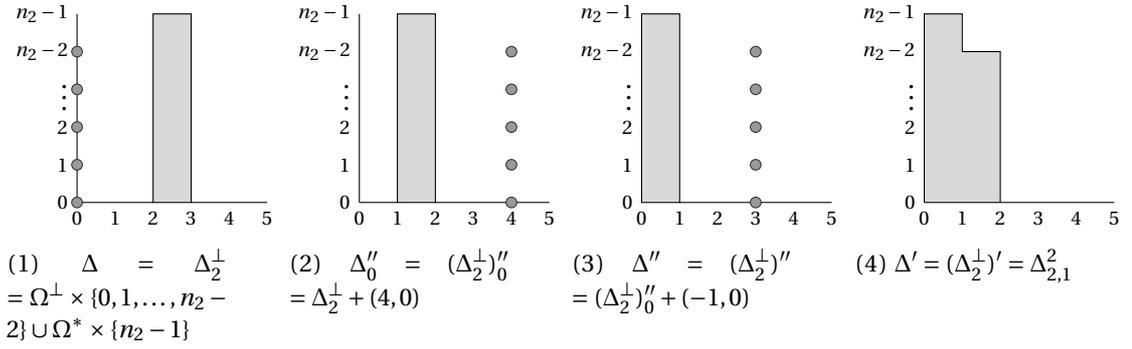
        The case $i=2$ can also be proved following the same arguments as above. It suffices to consider the symmetric situation, replace $P_1$ by $P_2$, $n_2$ by $n_1$ and use pairs $\left(\Delta,\Delta'\right)$ such that
        $$\left(\Delta,\Delta'\right)\in \left\{\left(\Delta_1,\Delta^1_{n_1-1,2}\right),\left(\Delta_1^\perp,\Delta^1_{n_1-1,2^h-2}\right),\left(\Delta_2,\begin{cases}
            \Delta^{2,\sigma}_{2,0}, & \text{when } j=2^h+1,\\
            \Delta^{3,\sigma}_{j,2}, & \text{otherwise.}
        \end{cases}\right),\left(\Delta_2^\perp,\Delta^{2,\sigma}_{2^h-2,2^h-2z+1}\right)\right\}.$$

We conclude with a last difference with respect to the proof of Proposition \ref{bivq+1}.

- The fact that $C_{\Delta'}^{P,J}$ is optimal follows from Corollary \ref{colpar} (1) when $\Delta$ is $\Delta_1$ or $\Delta_1^\perp$, Corollary \ref{colpar} (3) when $\Delta=\Delta_2$ and $j<n_{i'}-1$ and Corollary \ref{colpar} (2) when $\Delta$ equals $\Delta_2^\perp$ or $\Delta_2$ and $j=n_{i'}-1$.
\end{proof}

\begin{remark}\label{remsub}
Propositions \ref{bivq+1} and \ref{bivq+2} do not give an exhaustive list of the optimal $(r,\delta)$-codes one can find from subfield-subcodes of MCCs. These results are designed for providing $p^h$-ary optimal $(r,\delta)$-LRCs such that $r+\delta-1$ is either $p^h+1$ or $p^h+2$ and their lengths are a multiple of $r+\delta-1$, $r>1$ and $\delta>2$. Notice that these codes are new with respect to those given in the literature. The gcd-type conditions given in Proposition \ref{bivq+1} items (3) and (6) are stated to provide new parameters with respect to those obtained in \cite{SZW2019}. Moreover, excepting Proposition \ref{bivq+2}, items (5), with $j\neq n_{i'}-1$, and (7) (where $d\geq r+\delta)$, codes in both propositions have minimum distances $d\leq \min\{r+\delta,2\delta\}$, being new with respect to \cite{KWG2021}.
\end{remark}

\begin{exs}\label{examplessubfield}
In these examples, we give some new optimal LRCs obtained by applying Propositions \ref{bivq+1} and \ref{bivq+2}.
\begin{enumerate}
    \item Consider $(q,p^h,i,z,t,n_1,n_2)=(5^2,5,2,1,0,9,6)$, then by Proposition \ref{bivq+1} (3) one gets a $[54,25,6]_5$ optimal $(3,4)$-LRC.
    
    \item Consider $(q,p^h,i,z,n_1,n_2)=(7^2,7,2,2,17,8)$, then by Proposition \ref{bivq+1} (2) one gets a $[136,85,4]_7$ optimal $(5,4)$-LRC.
    
    \item Consider $(q,p^h,i,z,t,n_1,n_2)=(9^2,9,1,3,1,10,21)$, then by Proposition \ref{bivq+1} (3) one gets a $[210,143,8]_9$ optimal $(7,4)$-LRC.
    
    \item Consider $(q,p^h,i,n_1,n_2)=(2^4,4,1,6,15)$, then by Proposition \ref{bivq+2} (1) one gets a $[90,45,4]_4$ optimal $(3,4)$-LRC.
    
    \item Consider $(q,p^h,i,j,n_1,n_2)=(2^6,8,2,6,8,10)$, then by Proposition \ref{bivq+2} (6) one gets a $[80,19,20]_8$ optimal $(3,8)$-LRC.

    \item Consider $(q,p^h,i,z,n_1,n_2)=(2^6,8,1,3,10,10)$, then by Proposition \ref{bivq+2} (8) one gets a $[100,58,7]_8$ optimal $(7,4)$-LRC.

\end{enumerate}

Figure \ref{fig:exsubfield} shows the sets $\Delta$ introduced in Propositions \ref{bivq+1} and \ref{bivq+2} and used in the above examples. We make explicit the descomposition of the set $\Delta=\Delta_2=\{0,1,\dots,5\}\times\{0,1,8\}\cup \{(6,0)\}$ in Example \ref{examplessubfield} (5) as a union of minimal closed sets. Indeed, $(i,P_1,P_2,J)=(2,U_7\cup\{0\},U_9\cup\{0\},\emptyset)$ and $\Delta$ is the union of the following minimal closed sets:
\begin{equation*}
\begin{gathered}
    \Lambda_{(0,0)}=\{(0,0)\}, \quad \Lambda_{(1,0)}=\{(1,0)\}, \quad \Lambda_{(2,0)}=\{(2,0)\}, \quad \Lambda_{(3,0)}=\{(3,0)\}, \quad \Lambda_{(4,0)}=\{(4,0)\},\\
    \Lambda_{(5,0)}=\{(5,0)\}, \quad \Lambda_{(6,0)}=\{(6,0)\}, \quad \Lambda_{(0,1)}=\{(0,1),(0,8)\}, \quad \Lambda_{(1,1)}=\{(1,1),(1,8)\},\\
    \Lambda_{(2,1)}=\{(2,1),(2,8)\}, \quad \Lambda_{(3,1)}=\{(3,1),(3,8)\}, \quad \Lambda_{(4,1)}=\{(4,1),(4,8)\}, \quad \Lambda_{(5,1)}=\{(5,1),(5,8)\}.
\end{gathered}    
\end{equation*}

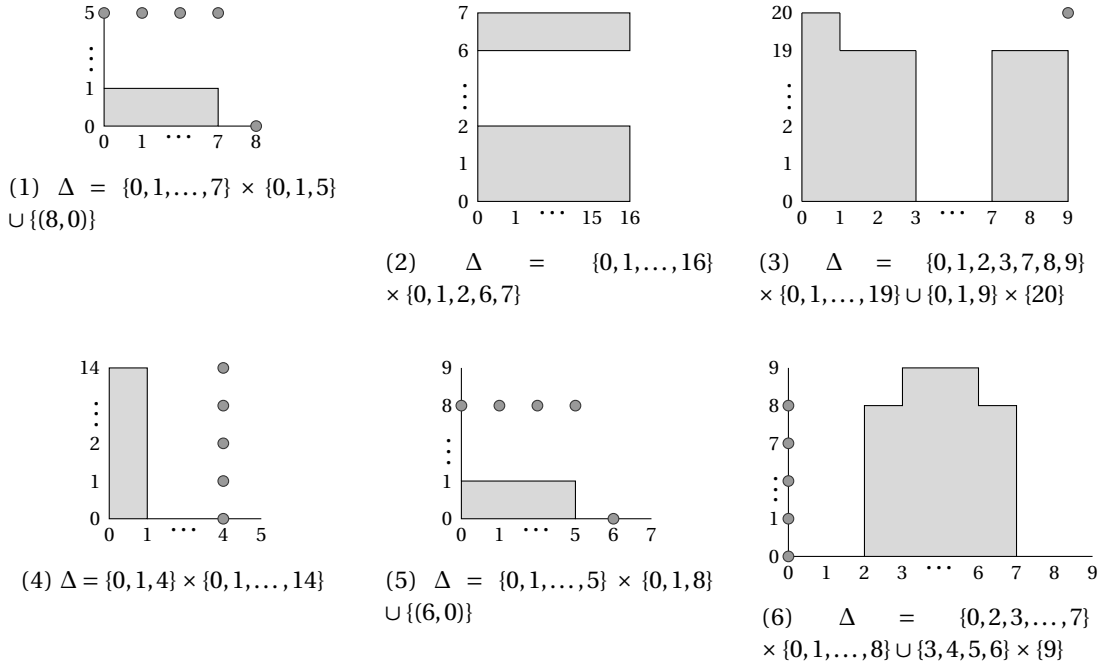
\begin{figure}[h]
    \centering
    \begin{multicols}{3}
    \centering
    \begin{subfigure}[b]{0.3\textwidth}
        \centering
        \begin{tikzpicture}[y=0.5cm, x=0.5cm,font=\normalsize]
        \filldraw[fill=gray!30] (0,0) rectangle (3,1);
        
        \draw (3,0) -- (4,0);
        \draw (0,1) -- (0,3);
        
        \filldraw[fill=black!40,draw=black!80] (4,0) circle (2pt);
        \filldraw[fill=black!40,draw=black!80] (0,3) circle (2pt);
        \filldraw[fill=black!40,draw=black!80] (1,3) circle (2pt);
        \filldraw[fill=black!40,draw=black!80] (2,3) circle (2pt);
        \filldraw[fill=black!40,draw=black!80] (3,3) circle (2pt);
        
        \node [below] at (0,0)  {\scriptsize$0$};
        \node [below] at (1,0) {\scriptsize$1$};
        \node [below] at (2,0) {$\dots$};
        \node [below] at (3,0) {\scriptsize$7$};
        \node [below] at (4,0) {\scriptsize$8$};
        \node [left] at (0,0) {\scriptsize$0$};
        \node [left] at (0,1) {\scriptsize$1$};
        \node [left] at (0,2) {$\vdots$};
        \node [left] at (0,3) {\scriptsize$5$};
        \end{tikzpicture}
        \caption{$\Delta=\{0,1,\dots,7\}\times\{0,1,5\} \break \cup\{(8,0)\}$}
    \end{subfigure}
    
    \columnbreak 
    \centering
    \begin{subfigure}[b]{0.3\textwidth}
        \centering
        \begin{tikzpicture}[y=0.5cm, x=0.5cm,font=\normalsize]
        \filldraw[fill=gray!30] (0,0) rectangle (4,2);
        \filldraw[fill=gray!30] (0,4) rectangle (4,5);
        
        \draw (0,2) -- (0,4);
        
        \node [below] at (0,0)  {\scriptsize$0$};
        \node [below] at (1,0) {\scriptsize$1$};
        \node [below] at (2,0) {$\dots$};
        \node [below] at (3,0) {\scriptsize$15$};
        \node [below] at (4,0) {\scriptsize$16$};
        \node [left] at (0,0) {\scriptsize$0$};
        \node [left] at (0,1) {\scriptsize$1$};
        \node [left] at (0,2) {\scriptsize$2$};
        \node [left] at (0,3) {$\vdots$};
        \node [left] at (0,4) {\scriptsize$6$};
        \node [left] at (0,5) {\scriptsize$7$};
        
        \end{tikzpicture}
        \caption{$\Delta=\{0,1,\dots,16\} \break \times\{0,1,2,6,7\}$}
    \end{subfigure}
    
    \columnbreak
    \centering
    \begin{subfigure}[b]{0.3\textwidth}
        \centering
        \begin{tikzpicture}[y=0.5cm, x=0.5cm,font=\normalsize]
        \fill[gray!30] (0,0) rectangle (3,4);
        \fill[gray!30] (0,4) rectangle (1,5);
        \filldraw[fill=gray!30] (5,0) rectangle (7,4);
        
        \draw (0,0) -- (0,5);
        \draw (0,0) -- (5,0);
        \draw (3,0) -- (3,4);
        \draw (1,4) -- (3,4);
        \draw (1,4) -- (1,5);
        \draw (0,5) -- (1,5);
        
        \filldraw[fill=black!40,draw=black!80] (7,5) circle (2pt);
        
        \node [below] at (0,0)  {\scriptsize$0$};
        \node [below] at (1,0) {\scriptsize$1$};
        \node [below] at (2,0) {\scriptsize$2$};
        \node [below] at (3,0) {\scriptsize$3$};
        \node [below] at (4,0) {$\dots$};
        \node [below] at (5,0) {\scriptsize$7$};
        \node [below] at (6,0) {\scriptsize$8$};
        \node [below] at (7,0) {\scriptsize$9$};
        \node [left] at (0,0) {\scriptsize$0$};
        \node [left] at (0,1) {\scriptsize$1$};
        \node [left] at (0,2) {\scriptsize$2$};
        \node [left] at (0,3) {$\vdots$};
        \node [left] at (0,4) {\scriptsize$19$};
        \node [left] at (0,5) {\scriptsize$20$};
        
        \end{tikzpicture}
        \caption{$\Delta=\{0,1,2,3,7,8,9\} \break \times\{0,1,\dots,19\}\cup\{0,1,9\}\times\{20\}$}
    \end{subfigure}
    \columnbreak
    \end{multicols}
    
    \centering
    \begin{multicols}{3}
    \centering
    \begin{subfigure}[b]{0.3\textwidth}
        \centering
        \begin{tikzpicture}[y=0.5cm, x=0.5cm,font=\normalsize]
        \filldraw[fill=gray!30] (0,0) rectangle (1,4);
        
        \draw (1,0) -- (4,0);
        
        \filldraw[fill=black!40,draw=black!80] (3,0) circle (2pt);
        \filldraw[fill=black!40,draw=black!80] (3,1) circle (2pt);
        \filldraw[fill=black!40,draw=black!80] (3,2) circle (2pt);
        \filldraw[fill=black!40,draw=black!80] (3,3) circle (2pt);
        \filldraw[fill=black!40,draw=black!80] (3,4) circle (2pt);
        
        \node [below] at (0,0)  {\scriptsize$0$};
        \node [below] at (1,0) {\scriptsize$1$};
        \node [below] at (2,0) {$\dots$};
        \node [below] at (3,0) {\scriptsize$4$};
        \node [below] at (4,0) {\scriptsize$5$};
        \node [left] at (0,0) {\scriptsize$0$};
        \node [left] at (0,1) {\scriptsize$1$};
        \node [left] at (0,2) {\scriptsize$2$};
        \node [left] at (0,3) {$\vdots$};
        \node [left] at (0,4) {\scriptsize$14$};
        
        \end{tikzpicture}
        \caption{$\Delta=\{0,1,4\}\times\{0,1,\dots,14\}$}
    \end{subfigure}

    \columnbreak
    \centering
    \begin{subfigure}[b]{0.3\textwidth}
        \centering
        \begin{tikzpicture}[y=0.5cm, x=0.5cm,font=\normalsize]
        \filldraw[fill=gray!30] (0,0) rectangle (3,1);
        
        \draw (3,0) -- (5,0);
        \draw (0,1) -- (0,4);
        
        \filldraw[fill=black!40,draw=black!80] (4,0) circle (2pt);
        \filldraw[fill=black!40,draw=black!80] (0,3) circle (2pt);
        \filldraw[fill=black!40,draw=black!80] (1,3) circle (2pt);
        \filldraw[fill=black!40,draw=black!80] (2,3) circle (2pt);
        \filldraw[fill=black!40,draw=black!80] (3,3) circle (2pt);
        
        \node [below] at (0,0)  {\scriptsize$0$};
        \node [below] at (1,0) {\scriptsize$1$};
        \node [below] at (2,0) {$\dots$};
        \node [below] at (3,0) {\scriptsize$5$};
        \node [below] at (4,0) {\scriptsize$6$};
        \node [below] at (5,0) {\scriptsize$7$};
        \node [left] at (0,0) {\scriptsize$0$};
        \node [left] at (0,1) {\scriptsize$1$};
        \node [left] at (0,2) {$\vdots$};
        \node [left] at (0,3) {\scriptsize$8$};
        \node [left] at (0,4) {\scriptsize$9$};
        
        \end{tikzpicture}
        \caption{$\Delta=\{0,1,\dots,5\}\times\{0,1,8\} \break \cup\{(6,0)\}$}
    \end{subfigure}

    \columnbreak
    \centering
    \begin{subfigure}[b]{0.3\textwidth}
        \centering
        \begin{tikzpicture}[y=0.5cm, x=0.5cm,font=\normalsize]
        \fill[gray!30] (2,0) rectangle (6,4);
        \fill[gray!30] (3,4) rectangle (5,5);
        
        \draw (0,0) -- (8,0);
        \draw (0,0) -- (0,5);
        \draw (2,0) -- (2,4) -- (3,4) -- (3,5) -- (5,5) -- (5,4) -- (6,4) -- (6,0);
        
        \filldraw[fill=black!40,draw=black!80] (0,0) circle (2pt);
        \filldraw[fill=black!40,draw=black!80] (0,1) circle (2pt);
        \filldraw[fill=black!40,draw=black!80] (0,2) circle (2pt);
        \filldraw[fill=black!40,draw=black!80] (0,3) circle (2pt);
        \filldraw[fill=black!40,draw=black!80] (0,4) circle (2pt);
        
        \node [below] at (0,0) {\scriptsize$0$};
        \node [below] at (1,0) {\scriptsize$1$};
        \node [below] at (2,0) {\scriptsize$2$};
        \node [below] at (3,0) {\scriptsize$3$};
        \node [below] at (4,0) {$\dots$};
        \node [below] at (5,0) {\scriptsize$6$};
        \node [below] at (6,0) {\scriptsize$7$};
        \node [below] at (7,0) {\scriptsize$8$};
        \node [below] at (8,0) {\scriptsize$9$};
        \node [left] at (0,0) {\scriptsize$0$};
        \node [left] at (0,1) {\scriptsize$1$};
        \node [left] at (0,2) {$\vdots$};
        \node [left] at (0,3) {\scriptsize$7$};
        \node [left] at (0,4) {\scriptsize$8$};
        \node [left] at (0,5) {\scriptsize$9$};
        
        \end{tikzpicture}
        \caption{$\Delta=\{0,2,3,\dots,7\} \break \times\{0,1,\dots,8\}\cup\{3,4,5,6\} \times \{9\}$}
    \end{subfigure}
    \end{multicols}

    \caption{Sets $\Delta$ considered in Examples \ref{examplessubfield}}
    \label{fig:exsubfield}
\end{figure}

\end{exs}

Now, we state our main results in this subsection which are Theorems \ref{thmparbivq+1} and \ref{thmparbivq+2}. These results follow directly from Propositions \ref{bivq+1} and \ref{bivq+2} and provide explicitly the parameters and $(r,\delta)$-localities of the new optimal LRCs we have obtained.

\begin{thm}\label{thmparbivq+1}
Let $\;\mathbb{F}_q$ be a finite field with $q=p^l$, $p$ being a prime number and $l$ a positive integer. Consider another positive integer $h$ such that $h$ divides $l$, $p^h\geq 4$ if $p=2$ ($p^h\geq 5$, otherwise) and assume $p^h+1\mid q-1$. Consider also nonnegative integers $z$ and $t$ satisfying $0\leq t<z \leq \floor*{\frac{p^h}{2}}-1$, $2t\geq \max\{0,4z-p^h-1\}$. Regard $\; \mathbb{F}_{p^h}$ as a subfield of $\;\mathbb{F}_q$.

Then, there exists an optimal $(r,\delta)$-LRC over $\;\mathbb{F}_{p^h}$ with the following parameters depending on two integer variables $n'$ and $a$:
$$\left[n,k,d\right]_{p^h}=\left[(p^h+1)n',(n'-1)(2z+1)+2a+1, p^h+1-2a\right]_{p^h}$$
and
$$(r,\delta)=(2z+1,p^h-2z+1),$$
whenever some of the following conditions hold:
\begin{enumerate}
    \item $n'\mid q-1$ and $a=z$.
    
    \item $n'-1\mid q-1$ and $a=z$.
    
    \item $n'-1\mid q-1$, $a=t$ and, if $p$ is odd, either $\gcd(n',p^h)\neq 1$ or $\gcd(n',p^h+1)\neq 1$.
\end{enumerate}

\medskip
Assume now that $p=2$ and consider a nonnegative integer $u$ and, if $u\geq 1$, a nonnegative integer $v$, satisfying $0\leq u\leq \frac{p^h}{2}-2$, $0\leq v<u$ and $2v+1\geq\max\{0,4u+1-p^h\}$.

Then, there exists an optimal $(r,\delta)$-LRC over $\;\mathbb{F}_{p^h}$ with the following parameters depending on two integer variables $n'$ and $a$:
$$\left[n,k,d\right]_{p^h}=\left[(p^h+1)n',(n'-1)(2u+2)+2a+2, p^h-2a\right]_{p^h}$$
and
$$(r,\delta)=(2u+2,p^h-2u),$$
whenever some of the following conditions hold:
\begin{enumerate}
    \item $n'\mid q-1$ and $a=u$.
    
    \item $n'-1\mid q-1$ and $a=u$.
    
    \item $n'-1\mid q-1$ and $a=v$.
\end{enumerate}

\end{thm}

\begin{thm}\label{thmparbivq+2}
Let $\;\mathbb{F}_q$ be a finite field with $q=2^l$, $l\geq 4$ being an even positive integer and $h=\frac{l}{2}$. Consider also a positive integer $z$ satisfying $\; 2\leq z\leq 3$, $2^h-2z+1\geq \max\{0,2^h-6\}$. Regard $\; \mathbb{F}_{2^h}$ as a subfield of $\;\mathbb{F}_q$.

Then, there exists an optimal $(r,\delta)$-LRC over $\;\mathbb{F}_{2^h}$ with the following parameters depending on the integer variables $n'$, $a$, $b$ and $c$:
$$\left[n,k,d\right]_{2^h}=\left[(2^h+2)n',a(n'-1)+b, 2h+3-b\right]_{2^h}$$
and
$$(r,\delta)=(a,c),$$
whenever some of the following conditions hold:
\begin{enumerate}

    \item $n'\mid q-1$ and $(a,b,c)=(3,3,2^h)$.
    
    \item $n'-1\mid q-1$ and $(a,b,c)=(3,3,2^h)$.
    
    \item $n'\mid q-1$ and $(a,b,c)=(2^h-1,2^h-1,4)$.

    \item $n'-1\mid q-1$ and $(a,b,c)=(2^h-1,2^h-1,4)$.

    \item $n'-1\mid q-1$ and $(a,b,c)=(2^h-1,2^h-2z+2,4)$.
\end{enumerate}

Finally, consider $n'$ and $j$ positive integers such that $j\leq n'-1$ and they satisfy some of the following conditions:

\begin{enumerate}

    \item $n'\mid 2^h-1$ and $j\geq\max\{1, n'-2^{h-1}\}$.
    
    \item $n'-1\mid 2^h-1$ and $\max\{1,n'-2^{h-1}\}\leq j < n'-1$.
    
    \item $n'-1\mid q-1$ and $j=n'-1$.
\end{enumerate}
Then, there exists an optimal $(r,\delta)$-LRC over $\;\mathbb{F}_{2^h}$ with parameters
$$\left[n,k,d\right]_{2^h}=\left[(2^h+2)n',3j+1, (2^h+2)(n'-j)\right]_{2^h}$$
and
$$(r,\delta)=(3,2^h).$$
\end{thm}

Table \ref{table:parbiv} shows parameters of some new optimal $(r,\delta)$-LRCs coming from subfield-subcodes deduced from Theorems \ref{thmparbivq+1} and \ref{thmparbivq+2}.

\begin{table}[h!]
\centering
\begin{tabular}{l c c c c c c c c} 
 \hline
 Item in Theorem & $p^h$ & $q$ & $n$ & $k$ & $d$ & $r$ & $\delta$ \\ [0.5ex] 
 \hline
 \ref{thmparbivq+1} (3) (for $(n',z,t)=(25,1,0)$) & 5 & 25 & $150=6\cdot25$ & 73 & 6 & 3 & 4 \\
 \ref{thmparbivq+1} (1) (for $(n',z)=(48,1)$) & 7 & 49 & $384=8\cdot48$ & 144 & 6 & 3 & 6 \\
 \ref{thmparbivq+1} (2) (for $(n',u)=(16,0)$) & 4 & 16 & $80=5\cdot16$ & 32 & 4 & 2 & 4 \\
 \ref{thmparbivq+1} (3) (for $(n',z,t)=(22,2,0)$) & 8 & 64 & $198=9\cdot22$ & 106 & 9 & 5 & 5 \\
 \ref{thmparbivq+2} (2) (for $(n',j)=(8,5)$) & 8 & 64 & $80=10\cdot8$ & 16 & 30 & 3 & 8 \\
 \ref{thmparbivq+2} (3) (for $n'=18$) & 4 & 256 & $108=6\cdot18$ & 144 & 6 & 3 & 6 \\ [1ex]
\end{tabular}
\caption{Optimal subfield-subcodes over $\mathbb{F}_{p^h}$}
\label{table:parbiv}
\end{table}

\subsection{Optimal $(r,\delta)$-LRCs coming from subfield-subcodes of multivariate MCCs}\label{opmultsubfield}

This section is devoted to extend Propositions \ref{bivq+1} and \ref{bivq+2} and Theorems \ref{thmparbivq+1} and \ref{thmparbivq+2} to the multivariate case. The corresponding versions are stated in the below Propositions \ref{multq+1} and \ref{multq+2}, and Theorems \ref{thmparmultq+1} and \ref{thmparmultq+2}. Their proofs run parallel to those given in the bivariate case and we omit them.

Keep the notation as in Section \ref{framework} and Subsection \ref{factssubfield}. Fix $j_0\in\{1,\dots,m\}$ and $S_1$, $S_2\subseteq\{1,\dots,m\}\backslash\{j_0\}$ such that $S_1\cup S_2=\{1,\dots,m\}\backslash\{j_0\}$ and $S_1\cap S_2=\emptyset$.

For our first construction, keep the notation as in the paragraphs before Lemma \ref{lemmaph+1} 
but changing $i$ by $j_0$. In particular consider nonnegative integers $z$ and $t$ (and when $p=2$) $u$ and $v$ as in those paragraphs. Denote
$$O_{z,t}:=\left\{t+1,t+2,\dots,z,p^h+1-z, p^h+2-z,\dots,p^h-t\right\}$$
and
$$O_{u,v}:=\left\{\frac{p^h}{2}-u,\frac{p^h}{2}-u+1,\dots, \frac{p^h}{2}-v-1, \frac{p^h}{2}+v+2,\frac{p^h}{2}+v+3,\dots, \frac{p^h}{2}+u+1\right\}.$$
Define
        $$\Delta_1:=\{0,1,\dots,n_1-1\}\times\cdots\times \{0,1,\dots,n_{j_0-1}-1\} \times \Omega_z\times \{0,1,\dots,n_{j_0+1}-1\} \times \cdots\times \{0,1,\dots,n_m-1\},$$
        $$\Delta_2:=\Delta_1\big\backslash \left\{(n_1-1,\dots,n_{j_0-1}-1,e_{j_0},n_{j_0+1}-1,\dots,n_m-1) \mid e_{j_0} \in O_{z,t}\right\},
        $$
        $$\Delta^*_1:=\{0,1,\dots,n_1-1\}\times\cdots\times \{0,1,\dots,n_{j_0-1}-1\} \times \Omega_u^*\times \{0,1,\dots,n_{j_0+1}-1\} \times \cdots\times \{0,1,\dots,n_m-1\}$$
and
        $$\Delta^*_2:=\Delta^*_1\big\backslash \left\{(n_1-1,\dots,n_{j_0-1}-1,e_{j_0},n_{j_0+1}-1,\dots,n_m-1) \mid e_{j_0} \in O_{u,v}\right\}.
        $$

\begin{prop}\label{multq+1}
Keep the notation as above where $\;\mathbb{F}_{p^h}$ is regarded as a subfield of $\;\mathbb{F}_{q=p^l}$ and $p^h+1\mid q-1$. Fixed $j_0$ and $P_{j_0}=U_{p^h+1}$, the set of $p^h+1$-th roots of unity, the following statements determine sets $P=P_1\times \cdots \times P_m$, $J$ and $\Delta$ such that the subfield-subcodes $S_\Delta^{P,J}$ over the field $\mathbb{F}_{p^h}$ are optimal $(r,\delta)$-LRCs:
\begin{enumerate}
    \item $P_j=U_{n_j}$ for some $n_j$ such that $n_j\mid q-1$ whenever $j\in S_1$ and when $j\in S_2$ $P_j=U_{n_j-1}\cup\{0\}$ for some $n_j$ such that $n_j-1\mid q-1$; $J=S_1\cup\{j_0\}$ and $\Delta=\Delta_1$, in which case
    $$(r,\delta)=(2z+1,p^h-2z+1).$$

    \item $S_1=\emptyset$, for all $j\in S_2$ $P_j=U_{n_j-1}\cup\{0\}$ for some $n_j$ such that $n_j-1\mid q-1$ and, if $p$ is odd, either $\gcd\left(\prod_{j\in \{1,\dots,m\}\backslash\{j_0\}}n_j,p^h\right)\neq 1$ or $\gcd\left(\prod_{j\in \{1,\dots,m\}\backslash\{j_0\}}n_j,p^h+1\right)\neq 1$; $J=\{j_0\}$ and $\Delta=\Delta_2$, in which case 
    $$(r,\delta)=(2z+1,p^h-2z+1).$$

    \item $P_j=U_{n_j}$ for some $n_j$ such that $n_j\mid q-1$ when $j\in S_1$ and when $j\in S_2$ $P_j=U_{n_j-1}\cup\{0\}$ for some $n_j$ such that $n_j-1\mid q-1$; $J=S_1\cup\{j_0\}$ and $\Delta=\Delta^*_1$, in which case
    $$(r,\delta)=(2u+2,p^h-2u).$$

    \item $S_1=\emptyset$ and for all $j\in S_2$ $P_j=U_{n_j-1}\cup\{0\}$ for some $n_j$ such that $n_j-1\mid q-1$; $J=\{j_0\}$ and $\Delta=\Delta^*_2$, in which case
    $$(r,\delta)=(2u+2,p^h-2u).$$
\end{enumerate}
\end{prop}

For the second construction, we use the notation as in the paragraph before Lemma \ref{lemmaph+2} but changing $i$ by $j_0$. Define
$$\Delta_1:=\{0,1,\dots,n_1-1\}\times\cdots\times \{0,1,\dots,n_{j_0-1}-1\} \times \Omega \times \{0,1,\dots,n_{j_0+1}-1\} \times \cdots\times \{0,1,\dots,n_m-1\},$$
    $$\Delta_2:=\Delta_1\big\backslash \left\{(n_1-1,\dots,n_{j_0-1}-1,e_{j_0},n_{j_0+1}-1,\dots,n_m-1) \mid e_{j_0} \in\left\{1,2^h\right\}\right\},$$
    $$\Delta^\perp_1:=\{0,1,\dots,n_1-1\}\times\cdots\times \{0,1,\dots,n_{j_0-1}-1\} \times \Omega^\perp\times \{0,1,\dots,n_{j_0+1}-1\} \times \cdots\times \{0,1,\dots,n_m-1\}$$
and
    $$\Delta^\perp_2:=\Delta^\perp_1\big\backslash \left\{(n_1-1,\dots,n_{j_0-1}-1,e_{j_0},n_{j_0+1}-1,\dots,n_m-1) \mid e_{j_0} \in \left. \begin{cases}
    \{0\}, & \text{when } z=2,\\
            \{0,2,2^h-1\}, & \text{otherwise.}
    \end{cases}\right\}
    \right\}.
    $$

\begin{prop}\label{multq+2}
Keep the notation as above where $\;\mathbb{F}_{2^h}$ is regarded as a subfield of $\;\mathbb{F}_{q=2^{2h}}$. Fixed $j_0$ and $P_{j_0}=U_{2^h+1}\cup\{0\}$, the set of $\; 2^h+1$-th roots of unity together with $0$, the following statements determine sets $P=P_1\times \cdots \times P_m$, $J$ and $\Delta$ such that the subfield-subcodes $S_\Delta^{P,J}$ over the field $\mathbb{F}_{2^h}$ are optimal $(r,\delta)$-LRCs:
\begin{enumerate}
    \item $P_j=U_{n_j}$ for some $n_j$ such that $n_j\mid q-1$ whenever $j\in S_1$ and when $j\in S_2$ $P_j=U_{n_j-1}\cup\{0\}$ for some $n_j$ such that $n_j-1\mid q-1$; $J=S_1$ and $\Delta=\Delta_1$, in which case 
    $$(r,\delta)=(3,2^h).$$

    \item $P_j=U_{n_j}$ for some $n_j$ such that $n_j\mid q-1$ whenever $j\in S_1$ and when $j\in S_2$ $P_j=U_{n_j-1}\cup\{0\}$ for some $n_j$ such that $n_j-1\mid q-1$; $J=S_1$ and $\Delta=\Delta_1^\perp$, in which case 
    $$(r,\delta)=(2^h-1,4).$$

    \item $S_1=\emptyset$ and for all $j\in S_2$ $P_j=U_{n_j-1}\cup\{0\}$ for some $n_j$ such that $n_j-1\mid q-1$; $J=\emptyset$ and $\Delta=\Delta_2$, in which case
    $$(r,\delta)=(3,2^h).$$

    \item $S_1=\emptyset$ and for all $j\in S_2$ $P_j=U_{n_j-1}\cup\{0\}$ for some $n_j$ such that $n_j-1\mid q-1$; $J=\emptyset$ and $\Delta=\Delta^\perp_2$, in which case
    $$(r,\delta)=(2^h-1,4).$$
\end{enumerate}
\end{prop}

\begin{remark}\label{remsub2}
As in the case of bivariate codes (see Remark \ref{remsub}), Propositions \ref{multq+1} and \ref{multq+2} impose conditions in order to obtain new families of optimal $(r,\delta)$-LRCs.
\end{remark}

Finally, we state our main results for the multivariate case. They are Theorem \ref{thmparmultq+1} (respectively, \ref{thmparmultq+2}) which give parameters and $(r,\delta)$-localities of the optimal $(r,\delta)$-LRCs we have obtained in Proposition \ref{multq+1} (respectively, \ref{multq+2}).

\begin{thm}\label{thmparmultq+1}
Let $\;\mathbb{F}_q$ be a finite field with $q=p^l$, $p$ being a prime number and $l$ a positive integer. Consider another positive integer $h$ such that $h$ divides $l$, $p^h\geq 4$ if $p=2$ ($p^h\geq 5$ otherwise) and assume $p^h+1\mid q-1$. Consider also nonnegative integers $z$ and $t$ satisfying $0\leq t<z \leq \floor*{\frac{p^h}{2}}-1$, $2t\geq \max\{0,4z-p^h-1\}$ and subsets $S_1$, $S_2\subseteq\{1,\dots,m-1\}$ such that $S_1\cup S_2=\{1,\dots,m-1\}$ and $S_1\cap S_2=\emptyset$. Regard $\; \mathbb{F}_{p^h}$ as a subfield of $\;\mathbb{F}_q$.

Then, there exists an optimal $(r,\delta)$-LRC over $\;\mathbb{F}_{p^h}$ with the following parameters depending on the integer variables $n_1,\dots,n_{m-1}$ and $a$: 
$$[n,k,d]_{p^h}=\left[(p^h+1)n_1\cdots n_{m-1},(2z+1)n_1\cdots n_{m-1}-a,p^h+1-2z+a\right]_{p^h}$$
and
$$(r,\delta)=(2z+1,p^h-2z+1),$$ 
whenever some of the following conditions hold:
\begin{enumerate}

    \item $n_j\mid q-1$ for all $j\in S_1$, $n_j-1\mid q-1$ for all $j\in S_2$ and $a=0$.
    
    \item $S_1=\emptyset$, $n_j-1\mid q-1$ for all $j\in S_2$, $a=2(z-t)$ and, if $p$ is odd, either $\gcd\left(n_1 \cdots n_{m-1},p^h\right)\neq 1$ or $\gcd\left(n_1 \cdots n_{m-1},p^h+1\right)\neq 1$.
\end{enumerate}

\medskip
Assume now that $p=2$ and consider a nonnegative integer $u$ and, if $u\geq 1$, a nonnegative integer $v$, satisfying $0\leq u\leq \frac{p^h}{2}-2$, $0\leq v<u$ and $2v+1\geq\max\{0,4u+1-p^h\}$.

Then, there exists an optimal $(r,\delta)$-LRC over $\;\mathbb{F}_{p^h}$ with parameters
$$[n,k,d]_{p^h}=\left[(p^h+1)n_1\cdots n_{m-1},(2u+2)n_1\cdots n_{m-1}-a,p^h-2u+a\right]_{p^h}$$
and
$$(r,\delta)=(2u+2,p^h-2u),$$
whenever some of the following conditions hold:
\begin{enumerate}

    \item $n_j\mid q-1$ for all $j\in S_1$, $n_j-1\mid q-1$ for all $j\in S_2$ and $a=0$.
    
    \item $S_1=\emptyset$, $n_j-1\mid q-1$ for all $j\in S_2$ and $a=2(u-v)$.
\end{enumerate}
\end{thm}

\begin{thm}\label{thmparmultq+2}
Let $\;\mathbb{F}_q$ be a finite field with $q=2^l$, $l\geq 4$ being an even positive integer and $h=\frac{l}{2}$. Consider also a positive integer $z$ satisfying $2\leq z\leq 3$, $2^h-2z+1\geq \max\left\{0,2^h-6\right\}$ and subsets $S_1$, $S_2\subseteq\{1,\dots,m-1\}$ such that $S_1\cup S_2=\{1,\dots,m-1\}$ and $S_1\cap S_2=\emptyset$. Regard $\; \mathbb{F}_{2^h}$ as a subfield of $\;\mathbb{F}_q$.

Then, there exists an optimal $(r,\delta)$-LRC over $\;\mathbb{F}_{2^h}$ with the following parameters depending on the integer variables $n_1,\dots,n_{m-1}$, $a$, $b$ and $c$:
$$[n,k,d]_{2^h}=\left[(2^h+2)n_1\cdots n_{m-1},an_1\cdots n_{m-1}-b,c+b\right]_{2^h}$$
and
$$(r,\delta)=(a,c),$$
whenever some of the following conditions hold:
\begin{enumerate}

    \item $n_j\mid q-1$ for all $j\in S_1$, $n_j-1\mid q-1$ for all $j\in S_2$ and $(a,b,c)=(3,0,2^h)$.

    \item $n_j\mid q-1$ for all $j\in S_1$, $n_j-1\mid q-1$ for all $j\in S_2$ and $(a,b,c)=(2^h-1,0,4)$.

    \item $S_1=\emptyset$, $n_j-1\mid q-1$ for all $j\in S_2$ and $(a,b,c)=(3,2,2^h)$.

    \item $S_1=\emptyset$, $n_j-1\mid q-1$ for all $j\in S_2$ and $(a,b,c)=(2^h-1,2z-3,4)$.
 
\end{enumerate}
\end{thm}

We finish this paper by giving, in Table \ref{table:parmult}, the parameters of some new optimal $(r,\delta)$-LRCs coming from subfield-subcodes deduced from Theorems \ref{thmparmultq+1} and \ref{thmparmultq+2}.

\begin{table}[h!]
\centering
\begin{tabular}{l c c c c c c c} 
 \hline
 Item in Theorem & $p^h$ & $q$ & $n$ & $k$ & $d$ & $r$ & $\delta$ \\ [0.5ex] 
 \hline
 \ref{thmparmultq+1} (1) (for $(m,z,t)=(3,1,0)$) & 5 & 625 & $480=6\cdot5\cdot16$ & 240 & 4 & 3 & 4 \\
 \ref{thmparmultq+1} (2) (for $(m,z,t)=(3,3,1)$) & 9 & 81 & $800=10\cdot8\cdot10$ & 556 & 8 & 7 & 4 \\
 \ref{thmparmultq+1} (2) (for $(m,z,t)=(4,1,0)$) & 4 & 16 & $320=5\cdot4\cdot4\cdot4$ & 190 & 5 & 3 & 3 \\
 \ref{thmparmultq+1} (2) (for $(m,u,v)=(3,2,0)$) & 8 & 64 & $720=9\cdot8\cdot10$ & 476 & 8 & 6 & 4 \\
 \ref{thmparmultq+2} (1) (for $m=4$) & 4 & 256 & $900=6\cdot5\cdot5\cdot6$ & 450 & 4 & 3 & 4 \\
 \ref{thmparmultq+2} (4) (for $(m,z)=(3,2)$) & 4 & 16 & $576=6\cdot6\cdot16$ & 287 & 5 & 3 & 4 \\ [1ex]
\end{tabular}
\caption{Optimal $(r,\delta)$-subfield-subcodes over $\mathbb{F}_{p^h}$}
\label{table:parmult}
\end{table}

\section*{Acknowledgements}

We thank H. H. López for indicating us the existence of monomial-Cartesian codes we had gone unnoticed. We named them \textit{zero-dimensional affine variety codes} in a previous version of this paper.

\clearpage
\bibliographystyle{plain}
\bibliography{biblio}
\hfill\break

\end{document}